\documentclass[a4paper,UKenglish,cleveref, autoref, thm-restate]{lipics-v2021}

\usepackage{tikz}
\usepackage{tikz-cd}
\usepackage{tikzit}

\tikzstyle{Black}=[fill=black, draw=black, shape=circle, inner sep=0, minimum size=6pt]
\tikzstyle{Small box}=[fill=white, draw=black, shape=chamfered rectangle, minimum size=8, chamfered rectangle corners=north east, chamfered rectangle xsep=2, chamfered rectangle ysep=2, tikzit shape=rectangle]
\tikzstyle{Small opbox}=[fill=white, draw=black, shape=chamfered rectangle, tikzit draw=white, tikzit fill={rgb,255: red,191; green,191; blue,191}, minimum size=8, chamfered rectangle corners=north west, chamfered rectangle xsep=2, chamfered rectangle ysep=2, tikzit shape=rectangle]
\tikzstyle{Red}=[fill=red, draw=black, shape=circle]
\tikzstyle{map small}=[fill=white, draw=black, shape=rounded rectangle, minimum width=23, minimum height=16, rounded rectangle west arc=none, tikzit shape=rectangle, tikzit draw=blue]
\tikzstyle{comap small}=[fill=white, draw=black, shape=rounded rectangle, minimum width=23, minimum height=16, rounded rectangle east arc=none, tikzit shape=rectangle, tikzit draw=red]
\tikzstyle{Medium box}=[fill=white, draw=black, shape=chamfered rectangle, minimum size=20, chamfered rectangle corners=north east, chamfered rectangle xsep=2, chamfered rectangle ysep=2, tikzit shape=rectangle]
\tikzstyle{Empty}=[fill=white, draw=black, shape=rectangle, dashed, minimum size=16, tikzit fill={rgb,255: red,191; green,191; blue,191}, tikzit draw=black]
\tikzstyle{White}=[fill=white, draw=black, shape=circle, inner sep=0, minimum size=6pt]
\tikzstyle{Vertex}=[fill=black, draw=black, shape=circle, inner sep=0, minimum size=3pt]
\tikzstyle{Edge}=[fill=white, draw=black, shape=rectangle, minimum size=8]
\tikzstyle{big edge}=[fill=white, draw=black, shape=rectangle, minimum size=26]
\tikzstyle{medium edge}=[fill=white, draw=black, shape=rectangle, minimum size=20]
\tikzstyle{big box}=[fill=white, draw=black, shape=chamfered rectangle, tikzit shape=rectangle, minimum size=28, chamfered rectangle corners=north east, chamfered rectangle xsep=2, chamfered rectangle ysep=2]

\tikzstyle{dashedge}=[dashed, -]
\tikzstyle{dasharrow}=[->, dashed]
\tikzstyle{twoheadarrow}=[<->]
\tikzstyle{new edge style 0}=[{|->}]
 \tikzset{	>=stealth',
auto,
	baseline=(current  bounding  box.center),
}
\tikzcdset{arrow style=tikz}
\usetikzlibrary{arrows}
\usepackage{calrsfs}
\DeclareMathAlphabet{\pazocal}{OMS}{zplm}{m}{n}

\usepackage{prooftree}
\usepackage{stmaryrd}

\title{On Doctrines and Cartesian Bicategories}

\author{Filippo Bonchi}{University of Pisa, Italy}{}{}{}\author{Alessio Santamaria}{University of Pisa, Italy}{alessio.santamaria@di.unipi.it}{https://orcid.org/0000-0001-7683-5221
}{}\author{Jens Seeber}{University of Pisa, Italy}{}{}{}\author{Pawe{\l} Soboci\'{n}ski}{Tallinn University of Technology, Estonia}{}{}{Supported by the ESF funded Estonian IT Academy research measure (project 2014-2020.4.05.19-0001) and the Estonian Research Council grant PRG1210.}

\authorrunning{F. Bonchi, A. Santamaria, J. Seeber and P. Soboci\'{n}ski } 

\Copyright{Filippo Bonchi, Alessio Santamaria, Jens Seeber, and Pawe{\l} Soboci\'{n}ski} 



\category{} 

\relatedversion{}

\funding{\emph{Filippo Bonchi, Alessio Santamaria}: Supported by the Ministero dell’Università e della Ricerca of Italy under Grant No. 201784YSZ5, PRIN2017 – ASPRA.}

\acknowledgements{We thank Pino Rosolini and Fabio Pasquali for discussions that improved our understanding of doctrines.}

\nolinenumbers 

\hideLIPIcs  

\EventEditors{John Q. Open and Joan R. Access}
\EventNoEds{2}
\EventLongTitle{42nd Conference on Very Important Topics (CVIT 2016)}
\EventShortTitle{CVIT 2016}
\EventAcronym{CVIT}
\EventYear{2016}
\EventDate{December 24--27, 2016}
\EventLocation{Little Whinging, United Kingdom}
\EventLogo{}
\SeriesVolume{42}
\ArticleNo{23}

\DeclareUnicodeCharacter{03B5}{\ensuremath{\epsilon}}

\newcommand{\op}[1]{#1^{\opposite}}
\DeclareMathOperator{\opposite}{op}
\DeclareMathOperator{\InfSL}{\mathsf{InfSL}}
\DeclareMathOperator{\Poset}{\mathsf{Poset}}
\DeclareMathOperator{\Hom}{Hom}
\DeclareMathOperator{\Rel}{\mathsf{Rel}}
\DeclareMathOperator{\Set}{\mathsf{Set}}
\DeclareMathOperator{\Map}{Map}
\DeclareMathOperator{\id}{id}

\newcommand{\ar}{\mathit{ar}}
\newcommand{\from}{\colon}
\newcommand{\seq}{\ensuremath{\mathrel{;}}}
\newcommand{\scale}[1]{\scalebox{0.7}{#1}}
\newcommand{\scalex}[1]{\scalebox{0.6}{#1}}
\newcommand{\stikzfig}[1]{\scale{\input{#1.tikz}}}
\newcommand{\sxtikzfig}[1]{\scalex{\input{#1.tikz}}}
\newcommand{\medtikzfig}[1]{\scalebox{0.5}{\input{#1.tikz}}}
\newcommand{\Cat}{\mathbb{C}}
\newcommand{\Bicat}{\mathcal{A}}

\newcommand{\CBC}{\mathsf{CBC}}
\newcommand{\EED}{\mathsf{EED}}
\newcommand{\EEDff}{\overline\EED} \newcommand{\LA}{\pazocal{L}}
\newcommand{\RA}{\pazocal{R}}
\newcommand{\A}{\mathbb A}
\newcommand{\B}{\mathbb B}
\newcommand{\C}{\mathbb C}
\newcommand{\D}{\mathbb D}
\newcommand{\E}{\mathbb E}
\newcommand{\pow}{\pazocal P}

\newcommand{\inv}[1]{{#1}^{-1}}
\renewcommand{\restriction}{\mathord{\upharpoonright}}
\renewcommand{\epsilon}{\varepsilon}

\newcommand{\Nat}{\mathbb N}

\newcommand{\sort}[2]{\colon (#1,#2)}
\newcommand{\Pred}{\mathbb{P}}

\newcommand{\Law}[1]{\mathsf{L}_{#1}} \newcommand{\fCB}[2]{\mathsf{CB}_{#1,#2}}

\tikzset{ihbase/.style={inner sep=0,circle,draw,fill=lightgray,minimum size=0.4em}}

\tikzset{ihblack/.style={ihbase,fill=black}}
\tikzset{ihwhite/.style={ihbase,fill=white}}

\tikzset{mat/.style={draw,fill=white,rectangle,node font=\scriptsize}}
\tikzset{ha/.style={mat, rectangle, rectangle }}
\tikzset{haop/.style={mat, rectangle, rectangle}}

\newcommand{\Bcomult}{\stikzfig{Copy}}

\newcommand{\Bcounit}{\stikzfig{Dis}}

\newcommand{\Bmult}{\stikzfig{Cocopy}}

\newcommand{\Bunit}{\stikzfig{Codis}}

\newcommand{\BcomultL}[1]{
\scalebox{0.7}{
\begin{tikzpicture}
	\begin{pgfonlayer}{nodelayer}
		\node [style=Black] (0) at (0, 0) {};
		\node [style=none] (1) at (-1.5, 0) {};
		\node [style=none] (2) at (1.5, 0.75) {};
		\node [style=none] (3) at (1.5, -0.75) {};
		\node [style=none] (4) at (-0.75, 0.25) {$#1$};
	\end{pgfonlayer}
	\begin{pgfonlayer}{edgelayer}
		\draw (1.center) to (0);
		\draw [in=183, out=90] (0) to (2.center);
		\draw [in=-180, out=-90] (0) to (3.center);
	\end{pgfonlayer}
\end{tikzpicture}

}
}

\newcommand{\BcounitL}[1]{
\scalebox{0.7}{
\begin{tikzpicture}
	\begin{pgfonlayer}{nodelayer}
		\node [style=none] (0) at (-0.75, 0) {};
		\node [style=Black] (1) at (0.75, 0) {};
		\node [style=none] (2) at (0, 0.25) {$#1$};
	\end{pgfonlayer}
	\begin{pgfonlayer}{edgelayer}
		\draw (0.center) to (1);
	\end{pgfonlayer}
\end{tikzpicture}

}
}

\newcommand{\dsymNetL}[2]{
\scalebox{0.7}{
\begin{tikzpicture}
	\begin{pgfonlayer}{nodelayer}
		\node [style=none] (0) at (-1, 0.75) {};
		\node [style=none] (1) at (-1, -0.75) {};
		\node [style=none] (2) at (1, 0.75) {};
		\node [style=none] (3) at (1, -0.75) {};
		\node [style=none] (4) at (-1.5, 0.75) {};
		\node [style=none] (5) at (-1.5, -0.75) {};
		\node [style=none] (6) at (1.5, 0.75) {};
		\node [style=none] (7) at (1.5, -0.75) {};
		\node [style=none] (8) at (-1.25, 1) {$#1$};
		\node [style=none] (9) at (1.25, 1) {$#2$};
		\node [style=none] (10) at (1.25, -0.5) {$#1$};
		\node [style=none] (11) at (-1.25, -0.5) {$#2$};
	\end{pgfonlayer}
	\begin{pgfonlayer}{edgelayer}
		\draw [in=180, out=0, looseness=0.75] (0.center) to (3.center);
		\draw [in=180, out=0, looseness=0.75] (1.center) to (2.center);
		\draw (4.center) to (0.center);
		\draw (5.center) to (1.center);
		\draw (2.center) to (6.center);
		\draw (3.center) to (7.center);
	\end{pgfonlayer}
\end{tikzpicture}
}
}

\newcommand{\idoneL}[1]{\tikz{
    \draw (0, 0) -- node[anchor=south] {\tiny{\ensuremath{#1}}} (1, 0);
}}

\newcommand{\idone}{
\tikz {
\draw (0, 0) -- (1, 0);
}
}

\newcommand\dsymNet{\stikzfig{Sym}}

\newcommand{\relation}[1]{
\scalebox{0.7}{\begin{tikzpicture}
	\begin{pgfonlayer}{nodelayer}
		\node [style=map small] (0) at (0, 0) {\normalsize ${#1}$};
		\node [style=none] (1) at (-1.25, 0) {};
		\node [style=none] (2) at (1.25, 0) {};
		\node [style=none] (3) at (-1, 0.25) {$n$};
	\end{pgfonlayer}
	\begin{pgfonlayer}{edgelayer}
		\draw (1.center) to (0);
		\draw (2.center) to (0);
	\end{pgfonlayer}
\end{tikzpicture}
}
}

\newcommand{\predicate}[1]{
\scalebox{0.7}{
\begin{tikzpicture}
	\begin{pgfonlayer}{nodelayer}
		\node [style=Small box] (0) at (0, 0) {\normalsize ${#1}$};
		\node [style=none] (1) at (-1.25, 0) {};
		\node [style=none] (3) at (-1, 0.25) {$n$};
	\end{pgfonlayer}
	\begin{pgfonlayer}{edgelayer}
		\draw (1.center) to (0);
	\end{pgfonlayer}
\end{tikzpicture}
}
}

\newcommand{\Zeronet}{\stikzfig{Empty}}

\renewcommand{\Theta}[1]{\left\llbracket #1 \right\rrbracket }

\def\invexcl{\rotatebox[origin=c]{180}{!}}

\theoremstyle{remark}
\newtheorem*{notation}{Notation}

\begin{document}

\maketitle
\begin{abstract}
We study the relationship between cartesian bicategories and a specialisation of Lawvere's hyperdoctrines, namely elementary existential doctrines. Both provide different ways of abstracting the structural properties of logical systems: the former in algebraic terms based on a string diagrammatic calculus, the latter in universal terms using the fundamental notion of adjoint functor. We prove that these two approaches are related by an adjunction, which can be strengthened to an equivalence by imposing further constraints on doctrines.
\end{abstract}

\section{Introduction}\label{sec:Intro}

In~\cite{lawvere_adjointness_1969,lawvere_diagonal_1969,lawvere_equality_1970} Lawvere introduced the notion of \emph{hyperdoctrine} in an effort to capture the \emph{universal} content of logical theories, and first-order logic in particular. Here, by universal, we intend by means of \emph{universal properties} in category theory.
The starting point is the notion of Lawvere theory~\cite{lawvere1963functorial}, the universal way of capturing the notion of algebraic theory -- where the universal property is that of cartesian categories, namely categories with finite products. In terms of logical content, Lawvere theories provide the notion of \emph{term}. Now, a hyperdoctrine is a certain contravariant functor $P$ from the Lawvere theory of terms to a posetal 2-category, e.g.\ lattices or Heyting algebras. The basic, high-level idea is that the functor takes us from terms to \emph{formulas}; more precisely, the objects of the Lawvere theories, which can be thought of as variable contexts, are taken to the Lindenbaum-Tarski algebra of formulas over these contexts. In this way, the concept of \emph{quantifier} can be captured by means of a universal property -- the existence of left-adjoints (existential quantification) and right-adjoints (universal quantification) to the image along $P$ of the projections.

In recent years, there has been a large number~\cite{BaezErbele-CategoriesInControl,DBLP:journals/pacmpl/BonchiHPSZ19,Bonchi2015,coecke2011interacting, Fong2015,DBLP:journals/corr/abs-2009-06836,Ghica2016,DBLP:conf/lics/MuroyaCG18,Piedeleu2021} of contributions that use string diagrams in order to model computational phenomena of different kinds. Typically, the languages come with an equational theory, which can be used to reason about systems via diagrammatic reasoning. Interestingly, the \emph{same} algebraic structures seem to appear in many different contexts, e.g. commutative monoids and comonoids, Frobenius algebras, Hopf algebras, etc. These applications, while using the language and tools of (monoidal) category theory, are of rather different nature than the more established ``universal'' approaches, such as Lawvere theories and hyperdoctrines sketched in the previous paragraph. A bridge between the universal and algebraic worlds is given by a theorem of Fox~\cite{fox_coalgebras_1976} that characterises cartesian categories as symmetric monoidal categories where each object is equipped with a well-behaved commutative comonoid structure. This means that any Lawvere theory can be seen concretely as a string diagrammatic language (see e.g.,~\cite{DBLP:journals/jlp/BonchiSZ18}). More recently, the notion of discrete cartesian restriction category was characterised in a similar way~\cite{Liberti2021}, with partial Frobenius algebras taking the place of commutative comonoids. This raises a natural question: can we capture the universal content of logical theories algebraically in a similar way? In other words, what are the ``Fox theorems'' for logic?

In this paper we turn our attention to the regular fragment of first-order logic with equality: formulas are built up from terms and the equality relation using the existential quantifier and conjunction. There has been much work on categorifying this fragment, notably the significant corpus of work on allegories~\cite{freyd1990categories}. More relevant to our story, we focus on the contrast between universal and algebraic approaches. A universal treatment, the notion of \emph{elementary existential doctrine}, was introduced in~\cite{maietti_quotient_2013}. The basic setup is the same as for Lawvere's hyperdoctrines, but one asks only for the left adjoints, which, as we have previously mentioned, are the universal explanation for existential quantifiers. On the algebraic side, the concept that stands out is that of Carboni and Walters' cartesian bicategories (of relations)~\cite{carboni1987cartesian}, which are symmetric monoidal categories where objects are equipped with a special Frobenius algebra and a lax-natural commutative comonoid structures. While Carboni and Walters emphasised the relational algebraic aspects, they were certainly aware of the logical connections. In fact, some recent works~\cite{GCQ,fong2018graphical,patterson2017knowledge} exploited various ramifications of the correspondence between cartesian bicategories and regular logic.

Our goal for this paper is a ``Fox theorem'' for the regular fragment, connecting the universal and the algebraic approaches. Our starting observation is that, given a cartesian bicategory $(\B, \otimes, I)$, one obtains an  elementary existential doctrine by restricting the hom functor $\Hom_\B(-,I) \colon \op\B \to \Set$ to the (cartesian) category of \emph{maps} of $\B$.
In Remark 2.5 of~\cite{maietti2015unifying}, it is mentioned that the other direction is also possible: given an elementary existential doctrine one can construct a cartesian bicategory. We explore the ramifications of this remark in detail. We show that these two translations are functorial and, actually, that they form an adjunction. More precisely, it turns out that the category of cartesian bicategories is a reflective subcategory of the category of doctrines.

The adjunction, however, is \emph{not} an equivalence. We prove this with a counterexample that captures the crux of the matter:
there are doctrines $P \colon \op\Cat \to \InfSL$ where the indexing categories of terms $\Cat$ are not tailored to the represented logics. In doctrines-as-logical-theories, roughly speaking, equality can come from two places: implicitly, from the indexing term category, and explicitly, via logical equivalence. Doctrines, therefore, have an additional degree of intensionality: doctrines that ``substantially represent'' the same logic may have distinct index categories and thus not be isomorphic. This issue does not arise in cartesian bicategories where the role of $\Cat$ is played by the subcategory of maps: maps are arrows satisfying certain properties, rather than given a priori in a fixed index category.

We conclude by observing that, by adding further constraints to the notion of elementary existential doctrine, namely comprehensive diagonals and the Rule of Unique Choice from \cite{maietti2015unifying}, it is possible to exclude such problematic doctrines. By doing so, we restrict the adjunction to an equivalence, thus obtaining a satisfactory ``Fox theorem''.

\begin{notation}
Given $f \colon X \to Y$ and $g \colon Y \to Z$ morphisms in some category, we denote their composite as $f \seq g$, $g \circ f$ or $gf$. We write $P_f$ for the action of a doctrine $P$ on a morphism $f$.
\end{notation}

 \section{Cartesian Categories}
The staring point of our exposition is the definition of cartesian category that, thanks to the results of Fox \cite{fox_coalgebras_1976},
can be given in the following form, which is particularly convenient for the purposes of this paper.
\begin{definition}\label{def:cartCat}
\	A cartesian category is a symmetric monoidal category $(\Cat,\otimes, I)$ where every object $X \in \Cat$ is equipped with morphisms
	\[
	\stikzfig{CopyX} \from X \to X \otimes X \quad \text{and} \quad \stikzfig{DisX} \from X \to I \quad \text{ such that}
	\]
\begin{enumerate}
		\item $\stikzfig{CopyX}$ and $\stikzfig{DisX}$ form a cocommutative comonoid, that is they satisfy
		\begin{align*}
		\sxtikzfig{Assoc} & \qquad
		\sxtikzfig{Comm} & \quad
		\sxtikzfig{Unit}
		\end{align*}
		\item Each morphism $f \from X \to Y$ is a comonoid homomorphism, that is
		\begin{align*}
		& \sxtikzfig{natcopy}
		& &\sxtikzfig{natdis}
		\end{align*}
		\item The choice of comonoid on every object is coherent with the monoidal structure in the sense that
		\begin{align*}
		& \sxtikzfig{CohDis}
		& & \sxtikzfig{CohCopy}
		\end{align*}
	\end{enumerate}
\end{definition}
Indeed,  given a category $\Cat$ with finite products, one can construct a monoidal category as in the definition above, by taking as monoidal product $\otimes$ the categorical product and as its unit $I$ the terminal object; for every object $X$, $\stikzfig{CopyX}$ is given by the pairing $\langle \id_X, \id_X \rangle \colon X \to X \otimes X$, hereafter denoted by $\Delta_X$, and $\stikzfig{DisX}$ by the unique morphism $!_X \colon X \to 1$. Conversely, given a symmetric monoidal category $(\Cat,\otimes, I)$  as in Definition \ref{def:cartCat}, $\otimes$ forms a categorical product where projections $\pi_X: X \times Y \to X$ are given as $id_X \otimes  \stikzfig{DisY}$ and the pairing $\langle f,g\rangle \colon X \to Y \otimes Z$ as  $\stikzfig{CopyX} ; (f \otimes g)$ for all arrows $f\colon X \to Y$ and $g\colon X \to Z$.

Hereafter, we will use $\times$ to denote both the cartesian product of sets and the categorical product in an arbitrary cartesian category. It is worth to remark that while $\Set$ (sets and functions) and $\Rel$ (sets and relations) are both cartesian categories, the categorical product in $\Set$ is indeed the cartesian product, while in $\Rel$ it is actually the disjoint union.

\begin{example}
Another cartesian category that will play an important role is $\Law{\Sigma}$, the \emph{Lawvere theory} \cite{lawvere1963functorial} generated by a cartesian signature $\Sigma$ (a set of symbols $f$ equipped with some arity $\ar(f)\in \Nat$). In $\Law{\Sigma}$, objects are natural numbers and arrows are tuples of terms over a countable set of variables $V=\{x_1,x_2, \dots\}$. More precisely, arrows from $n$ to $m$ are tuples $\langle t_1, \dots, t_m \rangle$ of sort $(n,m)$ as defined by the inference rules in the first line of Figure \ref{fig:sort}. It is easy to check that $\langle t_1, \dots, t_m \rangle$ has sort $(n,m)$ if each term $t_i$ has variables in $\{ x_1, \dots, x_m \}$. Composition is defined by (simultaneous) substitution: the composition of $\langle t_1, \dots, t_m \rangle \sort{n}{m}$ with $\langle s_1, \dots, s_l \rangle \sort{m}{l}$ is the tuple $\langle u_1, \dots, u_l \rangle \sort{n}{l} $ where $u_i= s_i[ ^{t_1 \dots t_m} /  _{x_1 \dots x_m} ]$ for all $i=1, \dots, l$.
One can readily check that $\Law{\Sigma}$ is a SMC having $(\Nat,+,0)$ as the monoid of objects, i.e., it is a  \emph{prop} (product and permutation category, see~\cite{Lack2004a,MacLane1965}). Identities $id_n$ and symmetries $\sigma_{n,m}$ are defined as expected; $\BcomultL{n} \colon n \to n+n$ is the tuple $\langle x_1, \dots, x_n,x_1, \dots, x_n\rangle $ thus acting as a \emph{duplicator} of variables;  $\BcounitL{n}\colon n \to 0$ is the empty tuple $\langle \rangle$, acting as a \emph{discharger}.
\end{example}

\begin{definition}
	A morphism of cartesian categories is a strict monoidal functor preserving the chosen comonoid structures.
\end{definition}

\begin{example}\label{ex:Lawmorphism}
Let $\Sigma_1$ be the cartesian signature consisting of a single symbol $f$ with arity~$1$, and $\Sigma_2$ be the signature with two symbols, $g_1$ and $g_2$, both of arity $1$.
Consider the corresponding Lawvere theories $\Law{\Sigma_1}$ and $\Law{\Sigma_2}$.
The assignment $f \mapsto g_1$ induces a morphism of cartesian categories, hereafter denoted by $F_1 \colon \Law{\Sigma_1} \to \Law{\Sigma_2}$. Similarly, let $F_2 \colon \Law{\Sigma_1} \to \Law{\Sigma_2}$ denote the morphism of cartesian categories where $f$ is mapped to $g_2$.
Finally, there is a unique morphism of cartesian categories $Q \colon \Law{\Sigma_2} \to \Law{\Sigma_1}$ mapping $g_1$ and $g_2$ to $f$.
\end{example}

\begin{figure}[t]
\[
\scalebox{0.72}{\ensuremath{
\begin{gathered}
\begin{prooftree}
{i\leq n}
\justifies
x_i \sort{n}{1}
\using {(V)}
\end{prooftree}
\quad
\begin{prooftree}
{f \in \Sigma \quad \ar(f)=m \quad \langle t_1, \dots t_m\rangle \sort{n}{m}}
\justifies
f\langle t_1,\dots, t_m\rangle \sort{n}{1}
\using {(\Sigma)}
\end{prooftree}
\quad
\begin{prooftree}
{t_1 \sort{n}{1}\quad \langle t_2,\dots, t_m \rangle \sort{n}{m-1}}
\justifies
\langle t_1, \dots, t_m\rangle \sort{n}{m}
\using {\langle\dots \rangle}
\end{prooftree}
\quad
\begin{prooftree}
\justifies
\langle \rangle \sort{n}{0}
\using {\langle \rangle}
\end{prooftree}
\\[1.5em]
\begin{prooftree}
{}
\justifies
\top \sort{n}{0}
\using {(\top)}
\end{prooftree}
\quad
\begin{prooftree}
{P \in \Pred \quad \ar(P)=m \quad \langle t_1, \dots t_m\rangle \sort{n}{m}}
\justifies
P\langle t_1,\dots, t_{m}\rangle \sort{n}{0}
\using {(\Pred)}
\end{prooftree}
\quad
\begin{prooftree}
 \quad \phi\sort{n+1}{0}
\justifies
 \exists x_{n+1}.\phi \sort{n}{0}
\using
{(\exists)}
\end{prooftree}
\quad
\begin{prooftree}
{\langle t_1,t_2\rangle \sort{n}{2}}
\justifies
t_1 = t_2 \sort{n}{0}
\using (=)
\end{prooftree} \quad
\begin{prooftree}
 \phi \sort{n}{0} \quad
 \psi \sort{n}{0}
\justifies
 \phi \wedge \psi \sort{n}{0}
\using (\wedge)
\end{prooftree}
\end{gathered}
}
}
\]
\caption{Sort inference rules}\label{fig:sort}
\end{figure}

\section{Elementary Existential Doctrines}\label{sec:eed}

Recall that an \emph{inf-semilattice} is a partially ordered set with all finite infima, including a top element $\top$. We denote the category of inf-semilattices and inf-preserving functions by $\InfSL$.

The following definition is taken almost verbatim from~\cite{maietti2015unifying}. The difference is that there the base category $\C$ only needs binary products, whereas we also require a terminal object.
\begin{definition}\label{def:eeh}
Let $\Cat$ be a cartesian category.
	An elementary existential doctrine is given by a functor $P \from \op{\Cat} \to \InfSL$
  that is:
\begin{itemize}
		\item \emph{Elementary}, namely for every object $A$ in $\Cat$ there is an element $\delta_A \in P(A \times A)$ such that
for every map $e = id_X \times \Delta_A \from X\times A \to X \times A \times A$,
			the function $P_e \from P(X \times A \times A) \to P(X \times A)$ has a left adjoint $\exists_{e} \colon P(X \times A) \to P(X \times A \times A)$ defined by the assignment
			\begin{equation}
				\exists_{e}(\alpha) = P_{\langle \pi_1, \pi_2 \rangle  \colon X \times A \times A \to X \times A}(\alpha) \wedge P_{\langle \pi_2, \pi_3 \rangle \colon X \times A \times A \to A \times A}(\delta_A). \label{eq:elementary2}
			\end{equation}
\item \emph{Existential}, namely for every $A_1, A_2 \in \Cat$ and projection $\pi_i \from A_1 \times A_2 \to A_i$ with $i \in \{1,2\}$,
		the function $P_{\pi_i} \from P(A_i) \to P(A_1 \times A_2)$ has a left-adjoint $\exists_{\pi_i}$ that satisfies
		\begin{itemize}
			\item the \emph{Beck-Chevalley condition}:
			for any projection $\pi \from X \times A \to A$ and any pullback
			\begin{equation}\label{eq:Beck-Chevalley pullback diagram}
			\begin{tikzcd}
			X' \ar[d,"f'"'] \arrow{r}{\pi'} & A' \ar[d,"f"] \\
			X \times A \arrow{r}{\pi} & A
			\end{tikzcd}
			\;
			\text{it holds that } \exists_{\pi'}(P_{f'}(\beta)) = P_{f}(\exists_{\pi}(\beta))
			\text{ for any } \beta \in P(X\times A) \text{.}
			\end{equation}
			\item \emph{Frobenius reciprocity}:
			for any projection $\pi \from X \times A \to A$, $\alpha \in P(A)$ and $\beta \in P(X)$, it holds that
			$
			\exists_{\pi}(P_{\pi}(\alpha) \wedge \beta) = \alpha \wedge \exists_{\pi}(\beta)
			$.
		\end{itemize}
	\end{itemize}
\end{definition}

\begin{remark}\label{rmk:BC}
    Taking $X$ in~\eqref{eq:elementary2} to be the terminal object of $\C$ , one obtains that the function $\exists_{\Delta_A} \from P(A) \to P(A \times A)$ given by
			\begin{equation}\label{eq:elementary1}
				\exists_{\Delta_A}(\alpha) = P_{\pi_1}(\alpha) \wedge \delta_A
			\end{equation}
			is left-adjoint to $P_{\Delta_A} \from P(A \times A) \to P(A)$. This condition, which appears in~\cite{maietti2015unifying}, is therefore redundant when $\Cat$ has terminal object.
\end{remark}

\begin{remark}\label{rem:Beck-Chevalley pullback diagram simple}
    In any cartesian category, the diagram below is a pullback of $f \colon A' \to A$ with a projection $\pi \colon X \times A \to A$. 
    \[
    \begin{tikzcd}
    X \times A' \ar[r,"\pi_2"] \ar[d,"\id_X \times f"'] & A' \ar[d,"f"] \\
    X \times A \ar[r,"\pi"] & A
	\end{tikzcd}
    \]
    Therefore, given a functor $P \colon \op{\C} \to \InfSL$, to check that $P$ satisfies the Beck-Chevalley condition it suffices to show that~\eqref{eq:Beck-Chevalley pullback diagram} holds for $X'\!=X \times A'$, $\pi'=\pi_2$ and $f'=\id_X \times f$.
\end{remark}

The contravariant powerset $\pow$, while often seen as an endofunctor on $\Set$, can also be seen as a functor
$\pow \colon \op\Set \to \InfSL$. It is the classical example of an elementary existential doctrine. Recall that, for $f \colon Y \to X$ in $\Set$, $\pow_f \colon \pow(X) \to \pow(Y)$ is $\pow_f (Z) = \{ y \in Y \mid f(y) \in Z \}$. Given a projection $\pi \colon X \times A \to A$, $\pow_\pi$ and its left adjoint $\exists_\pi$ are as follows:
\[
    \begin{tikzcd}[row sep=0em]
    \pow(A) \ar[r,"\pow_\pi"] & \pow(X \times A) \\
    B \ar[r,|->] & \{ (x,b) \in X \times A \mid b \in B  \}
    \end{tikzcd}
        \qquad
    \begin{tikzcd}[row sep=0em,column sep=2em]
    \pow(X \times A) \ar[r,"\exists_{\pi}"] & \pow(A) \\
    S \ar[r,|->] & \{ a \in A \mid \exists x \in X \ldotp (x,a) \in S \}
    \end{tikzcd}
\]
For every set $A$, $\delta_A$ is fixed to be $\{ (a, a) \mid a \in A \} \in \pow(A \times A)$. With this information, it is easy to check that $\pow \colon \op\Set \to \InfSL$ satisfies the conditions in Definition \ref{def:eeh}. For the convenience of the reader, we work out the details in Appendix \ref{appendix:eed}.

\begin{example}\label{example:Lindenbaum-Tarski eed}
Let $\Sigma$ and $\Pred$ be signatures of function and predicate symbols respectively. The regular fragment of first order logic consists of formulas built from conjunction $\wedge$, true $\top$, existential quantification $\exists$, equality $t_1=t_2$ and atoms $P\langle t_1, \dots t_m \rangle$ where $P\in \Pred$ and $t_i$ are terms over $\Sigma$. Formulas are sorted according to the rules in Figure \ref{fig:sort}: $\phi$ has sort $(n,0)$ if the free variables of $\phi$ are in $\{x_1, \dots, x_n\}$.

The indexed Lindenbaum-Tarski algebra functor $LT \colon \op{\Law{\Sigma}} \to \InfSL$ assigns to each $n \in \Nat$ the set of formulas of sort $(n,0)$ modulo logical equivalence (defined in the usual way). These form a semilattice with top given by $\top$ and meet by $\wedge$, where $\phi \leq \psi$ if and only if $\psi$ is a logical consequence of $\phi$.
To the arrow $\langle t_1, \dots t_m\rangle \colon n \to m$, $LT$ assigns the substitution mapping each $\phi \sort{m}{0}$ to $\phi [^{t_1, \dots t_m}/_{x_1, \dots x_m}] \sort{n}{0}$. In particular, for the projection $\pi \colon n+m \to n$ (that is $\langle x_1, \dots, x_n\rangle \colon n+m \to n$) $LT_{\pi}$ maps a formula of sort $(n,0)$ to the same formula but with sort $(n+m,0)$\footnote{The second projection $\pi\colon n+m \to m$ requires the reindexing of variables. We do not discuss this case in order to keep the presentation simpler.}. Its left adjoint $\exists_\pi$ maps a formula $\phi \sort{n+m}{0}$ to the formula $\exists x_{n+1}. \dots \exists x_{n+m}.\, \phi \sort{n}{0}$. The Beck-Chevalley condition asserts that ``substitution commutes with quantification'': if $\phi$ is a formula with at most $x_{n+1}$ free variables and $t$ is a term that does not contain $x_{n+1}$, then $\exists x_{n+1}.\, (\phi[^t / _{x_i}]) = (\exists x_{n+1}.\, \phi )[^t / _{x_i}]$. Frobenius reciprocity states that $\exists x_i.\, (\phi \wedge \psi )= \phi \wedge (\exists x_i.\, \psi )$ if $x_i$ is not a free variable of $\phi$. For all $n \in \Nat$,
$\delta_n$ is the formula $(x_1 = x_{n+1}) \wedge (x_2=x_{n+2}) \wedge \dots \wedge (x_n=x_{n+n})$.

\end{example}

\begin{definition}[{Cf.\ \cite{maietti2015unifying}}]\label{def:category EED}
	The category $\EED$ consists of the following data.
	\begin{itemize}
		\item Objects are elementary existential doctrines $P \colon \op\C \to \InfSL$.
		\item Morphisms from $P \colon \op\C \to \InfSL$ to $R \colon \op\D \to \InfSL$ are pairs $(F,b)$, where $F \colon \C \to \D$ is a strict cartesian functor while $b \colon P \to R  \circ \op F$ is a natural transformation
		
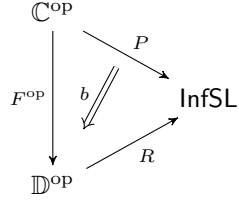
\[
		\begin{tikzcd}[column sep=3em]
		\op\C \ar[dr,"P"] \ar[dr,draw=none,swap,""name=P] \ar[dd,"\op F"']  \\
		& \InfSL \\
		\op\D \ar[ur,swap,"R"] \ar[ur,draw=none,""name=R,pos=.445]
		\arrow[Rightarrow,from=P,"b"',shorten >=1.5em,pos=0.4]
		\end{tikzcd}
		\]
		that preserves equalities and existential quantifiers, that is $b_{A \times A} (\delta^P_A) = \delta^R_{F(A)}$ for all $A$ in $\C$ and, for any projection $\pi \colon X \times A \to A$ in $\C$, the following diagram commutes.
		\begin{equation}\label{eq:commutative square b preserves existential quantifier}
		\begin{tikzcd}
		P(X \times A) \ar[r,"\exists^P_\pi"] \ar[d,"b_{X \times A}"'] & P(A) \ar[d,"b_A"] \\
		RF(X \times A) \ar[r,"\exists^R_{F(\pi)}"'] & RF(A)
		\end{tikzcd}
		\end{equation}
\item Composition of $(F,b) \colon P \to R$ as above with $(G,c) \colon R \to S$ is given by $(GF,cF \circ b)$.
	\end{itemize}
\end{definition}

\begin{remark}
    In~\cite{maietti2015unifying}, morphisms of elementary existential doctrines $P \to R$ are pairs $(F,b)$ where $F \colon \C \to \D$ is a functor between the base categories that preserves binary products merely up to isomorphism, so $\begin{tikzcd}[cramped, sep=small]
    F(X) & F(X \times Y) \ar[l,"F(\pi_1)"'] \ar[r,"F(\pi_2)"] & F(Y)
    \end{tikzcd}$
    is a product diagram of $F(X)$ and $F(Y)$ in $\D$ but it might not coincide with the \emph{chosen} product of $\D$. For this reason, preservation of equality for $b$ in~\cite{maietti2015unifying} means that $b_{A \times A}(\delta^P_A) = R_{\langle F(\pi_1), F(\pi_2) \rangle} (\delta^R_{F(A)})$.
\end{remark}

\begin{remark}\label{rem:precomposition of doctrine with cartesian functor}
	Let $P \colon \op\D \to \InfSL$ be an elementary existential doctrine and $F \colon \C \to \D$ a strict cartesian functor. Then $P \circ \op F \colon \op\C \to \InfSL$ is again an elementary existential doctrine, where $\delta^{P \op F}_A = \delta^P_{F(A)}$ for all $A$ in $\B$ and, for very $\pi_i \colon X_1 \times X_2 \to X_i$, $\exists^{P \op F}_{\pi_i} = \exists^P_{F(\pi_i)}$.
\end{remark}

\section{Cartesian Bicategories}\label{sec:cartesianbi}
In this section we recall from~\cite{carboni1987cartesian} definitions and properties of cartesian bicategories.

\begin{definition}\label{def:cartBicat}
	A cartesian bicategory is a symmetric monoidal category $(\B,\otimes, I)$ enriched over the category of posets (that is every hom-set is a partial order and both composition and $\otimes$ are monotonous operations) where every object $X \in \B$ is equipped with morphisms
	\[
	\stikzfig{CopyX} \from X \to X \otimes X \quad \text{and} \quad \stikzfig{DisX} \from X \to I \quad \text{ such that}
	\]
\begin{enumerate}
		\item $\stikzfig{CopyX}$ and $\stikzfig{DisX}$ form a cocommutative comonoid, as in Definition \ref{def:cartCat}.1
\item\label{monoid} $\stikzfig{CopyX}$ and $\stikzfig{DisX}$ have right adjoints $\stikzfig{CocopyX}$ and $\stikzfig{CodisX}$ respectively, that is
\[
        \medtikzfig{UnitCopy} \qquad
		\medtikzfig{CounitCopy} \qquad \medtikzfig{UnitDis}  \qquad
		\medtikzfig{CounitDis}
        \]
		\item\label{frob} The Frobenius law holds, that is
		\[
		\sxtikzfig{Frob}
		\]
		\item\label{lax} Each morphism $R \from X \to Y$ is a \emph{lax} comonoid homomorphism, that is
		\begin{align*}
		& \sxtikzfig{LaxCopy}
		& &\sxtikzfig{LaxDis}
		\end{align*}
		\item The choice of comonoid is coherent with the monoidal structure\footnote{In the original definition of~\cite{carboni1987cartesian} this property is replaced by requiring the uniqueness of the comonoid/monoid. However, as suggested in~\cite{patterson2017knowledge}, coherence seems to be the property of primary interest.}, as in Definition \ref{def:cartCat}.3.
\end{enumerate}
\end{definition}

The archetypal example of a cartesian bicategory is the category of sets and relations $\Rel$, with cartesian product of sets as monoidal product and $1=\{\bullet\}$ as unit $I$. To be precise, $\Rel$ has sets as objects and relations $R \subseteq X \times Y$ as arrows $X \to Y$. Composition and monoidal product are defined as expected:
$R \seq S=\{(x,z) \,| \, \exists y \text{ s.t.\ } (x,y)\in R\text{ and } (y,z)\in S\}$ and
$R\otimes S =\{\big((x_1,x_2)\,, \, (y_1,y_2)\big) \,|\, (x_1,y_1)\in R \text{ and } (x_2,y_2)\in S \}$.
 For each set $X$, the comonoid structure is given by the diagonal function $\Delta_X\colon X\to X\times X$ and the unique function $!_X\colon X\to 1$, considered as relations, that is $\stikzfig{CopyX} = \{\big(x, (x,x) \big) \, |\, x\in X\}$ and $\stikzfig{DisX}=\{(x,\bullet) \,|\, x\in X\}$. Their right adjoints are given by their opposite relations: $\stikzfig{CocopyX}=\{\big((x,x),x \big) \, |\, x\in X\}$ and $\stikzfig{CodisX}=\{(\bullet,x) \,|\, x\in X\}$.
 Following the analogy with $\Rel$, we will often call `relations' arbitrary morphisms of a cartesian bicategory.

One of the fundamental properties of cartesian bicategories that follows from the existence of right adjoints (Property~\ref{monoid} in Definition~\ref{def:cartBicat}) is that
every local poset $\Hom_{\B}(X,Y)$ allows to take the intersection of relations and has a top element.

\begin{lemma}\label{lemma:homsets in CBC are InfSLattices}Let $\B$ be a cartesian bicategory and $X,Y \in \B$. The poset $\Hom_{\B}(X,Y)$ has a top element given by $\stikzfig{Top}$ and the meet of relations
  $R,S \from X \to Y$ is given by
  \[ \stikzfig{Meet} \]
\end{lemma}

The Frobenius law (Property~\ref{frob}) gives a compact closed structure -- in other words, it allows us
to bend wires around. The cup of this compact closed structure is $\stikzfig{Cup}$, the cap analogously $\stikzfig{Cap}$ and the Frobenius law implies
the snake equations:
\begin{equation}\label{eq:snake}
  \sxtikzfig{Snake}
\end{equation}
To obtain an intuition for the lax comonoid homomorphism condition (Property~\ref{lax}), it is useful to spell out its meaning in $\Rel$: in the first inequality, the left and the right-hand side are, respectively, the relations $\{\big(x,(y,y)\big) \,| \,  (x,y)\in R \}$  and $\{\big(x,(y,z)\big) \,| \,  (x,y)\in R \text{ and } (x,z) \in R \}$, while in the second inequality, they are the relations
$\{(x,\bullet) \,| \, \exists y\in Y \text{ s.t. } (x,y)\in R \}$ and  $\{(x,\bullet) \,|\, x \in X  \}$.
It is immediate to see that the two left-to-right inclusions hold for any relation $R\subseteq X \times Y$, while the right-to-left inclusions hold for exactly those relations that are (graphs of) functions: the inclusions specify, respectively,
single-valuedness and totality.

\begin{definition}
  A map in a cartesian bicategory is an arrow $f$ that is a comonoid homomorphism, i.e.\ for which the equalities in Definition~\ref{def:cartCat}.2 hold. For $\B$ a cartesian bicategory, its category of maps, $\Map(\B)$, has the same objects of $\B$ and as morphisms the maps of $\B$.
\end{definition}
\begin{lemma} \label{lemma:MapRightAdj} A morphism $\stikzfig{Map}$ is a map if and only if it has a right adjoint -- a morphism $R$ such that
  $\stikzfig{SVAdj}$ and $\stikzfig{EntireAdj}$.
\end{lemma}
As expected, maps in $\Rel$ are precisely functions. Thus, $\Map(\Rel)$ is exactly $\Set$. For maps it makes sense to imagine a flow of information from left to right. We will therefore draw $\stikzfig{Map}$ to denote a map $f$.  Note that we use lower-case letters for maps and upper-case for arbitrary morphisms.
Since $\stikzfig{CopyX}$ and $\stikzfig{DisX}$ are maps by Lemma~\ref{lemma:MapRightAdj} and since by definition, every map respects them, then
$\Map(\B)$, which inherits the monoidal product from $\B$, has the structure of a cartesian category (Definition~\ref{def:cartCat}).

\begin{lemma}\label{lem:MapCart}
  For a cartesian bicategory $\B$, the monoidal product $\otimes$ is a product on $\Map(\B)$, and the monoidal unit $I$ is terminal.
In other words, $(\Map(\B),\, \otimes,\, I)$ is a cartesian category.
\end{lemma}
\begin{definition}
	A morphism of cartesian bicategories is a strict monoidal functor that preserves the ordering, the chosen monoid and the comonoid structures. Cartesian bicategories and their morphisms form a category $\CBC$.
\end{definition}

\begin{proposition}\label{prop:restriction of a morphism in CBC to maps}
   Let $F \colon \A \to \B$ be a morphism  of cartesian bicategories. Then restricting its domain to $\Map(\A)$ yields a strict cartesian functor $F\restriction_{\Map(\A)} \colon \Map(\A) \to \Map(\B)$.
\end{proposition}

The remaining sections focus on the relationship between cartesian bicategories and elementary existential doctrines. First, a little taste of the similarity between them.

\begin{figure}
\begin{tabular}{c}
\scalebox{0.96}{
\[
\scalebox{0.8}{
\begin{prooftree}
{}
\justifies
\Zeronet \colon 0 \to 0
\using {id_0}
\end{prooftree}
\quad\raisebox{3pt}{
\begin{prooftree}
{}
\justifies
\idone \colon 1 \to 1
\using {id_1}
\end{prooftree}}
\quad
\begin{prooftree}
{}
\justifies
\raisebox{1pt}{\scalebox{0.5}{\dsymNet}} \colon 2 \to 2
\using {\sigma_{1,1}}
\end{prooftree}
\quad
\raisebox{2pt}{
\begin{prooftree}
{}
\justifies
 \raisebox{1pt}{\scalebox{0.5}{\Bcounit}} \colon 1 \to 0
 \using {!_1}
\end{prooftree}}
\quad
\begin{prooftree}
{}
\justifies
\raisebox{1pt}{\scalebox{0.5}{\Bcomult}} \colon 1 \to 2
\using {\Delta_1}
\end{prooftree}
\quad
\begin{prooftree}
{}
\justifies
\raisebox{1pt}{\scalebox{0.5}{ \Bunit}} \colon 0 \to 1
 \using{\invexcl_1}
\end{prooftree}
\quad
{
\begin{prooftree}
{}
\justifies
\raisebox{1pt}{\scalebox{0.5}{\Bmult}} \colon 2 \to 1
\using{\nabla_1}
\end{prooftree}}
}
\]
}
\\
\\
\scalebox{0.96}{
\[
\scalebox{0.8}{
\begin{prooftree}
{f\in \Sigma \quad \ar(f)=n}
\justifies
\relation{f}\colon n \to 1
\using {\Sigma}
\end{prooftree}
\quad
\begin{prooftree}
{P\in \Pred \quad \ar(f)=n}
\justifies
\predicate{P}\colon n \to 0
\using{\Pred}
\end{prooftree}
\quad
\raisebox{2pt}{
\begin{prooftree}
{
\stikzfig{c} \colon n \to m \quad
\stikzfig{d}
\colon m \to o}
\justifies
\stikzfig{cd} \colon n \to o
\using {\, ; }
\end{prooftree}}
\quad
\begin{prooftree}
{
\stikzfig{c}
\colon n \to m \quad
\stikzfig{dop}
\colon o \to p}
\justifies
{\raisebox{-1pt}{$\stikzfig{cPARd}\colon n+o \to m+p$}}
\using {\, \otimes}
\end{prooftree}
}
\]
}
\\
\hline
\\
\scalebox{0.96}{
\[
\scalebox{1}{\ensuremath{
    \begin{gathered}
    \medtikzfig{Assocsimple} \qquad  \medtikzfig{Commsimple}  \qquad  \medtikzfig{Unitsimple} \\[1em]
    \medtikzfig{UnitCopySimple} \qquad \medtikzfig{CounitCopySimple} \qquad
    \medtikzfig{UnitDisSimple} \qquad \medtikzfig{CounitDisSimple} \\[1em]
    \medtikzfig{FrobNoX} \\[1em]
    \medtikzfig{natcopyAxiom} \qquad \medtikzfig{natdisAxiom}
     \qquad \qquad \qquad \qquad \medtikzfig{LaxCopySimple} \qquad \medtikzfig{LaxDisSimple}
    \end{gathered}
    }
    }
    \]
    }
\end{tabular}

\caption{Sort inference rules (top) and axioms (bottom) for $\fCB{\Sigma}{\Pred}$.}\label{fig:stringdiagrams}
\end{figure}

\begin{example}\label{ex:CBSP}
In Example \ref{example:Lindenbaum-Tarski eed}, we outlined the Lindenbaum-Tarski doctrine for the regular fragment of first order logic. We now introduce a cartesian bicategory, denoted by $\fCB{\Sigma}{\Pred}$, that provides a string diagrammatic calculus for this fragment. Like Lawvere theories, $\fCB{\Sigma}{\Pred}$ has $(\Nat,+,0)$ as monoid of objects. Arrows are (equivalence classes of) string diagrams~\cite{selinger2010survey} generated according to the rules in Figure~\ref{fig:stringdiagrams} (top).
For all $n,m\in \Nat$, identities $\idoneL{n}\colon n \to n$, symmetries $\scalebox{0.75}{\dsymNetL{n}{m}}\colon n+m \to m+n$, duplicators $\BcomultL{n} \colon n \to n+n$, dischargers $\BcounitL{n}\colon n \to 0$ and their adjoints
can be constructed from the basic diagrams in the first row.
Observe that function symbols $f$ with arity $n$ are depicted $\relation{f}$ with coarity $1$, while predicate symbols
$P\in \Pred$ with arity $n$ as $\predicate{P}$ with coarity $0$.

Figure~\ref{fig:stringdiagrams} (bottom) illustrates the axioms for the calculus: in the last row, function symbols $\relation{f}\colon n \to 1$ are forced to be comonoid homomorphism,  while predicate symbols $\predicate{P}\colon n \to 0$ are just lax; the first three rules impose properties 1, 2 and 3 of Definition~\ref{def:cartBicat} on the generating (co)monoid (the laws for arbitrary $n$ follow from these). Let $\leq_{\Sigma, \Pred}$ be the precongruence with respect to  $;$ and $\otimes$ generated by the axioms and $=_{\Sigma,\Pred}$ the corresponding equivalence, i.e., $\leq_{\Sigma, \Pred} \cap \geq_{\Sigma, \Pred}$. Now $\fCB{\Sigma}{\Pred}[n,m]$ is exactly the set of $=_{\Sigma,\Pred}$-equivalence classes of diagrams $d\colon n \to m$ ordered by $\leq_{\Sigma, \Pred}$. Simple inductions suffice to check that $\fCB{\Sigma}{\Pred}$ is indeed a cartesian bicategory and that, moreover, $\Map(\fCB{\Sigma}{\Pred})$ is isomorphic to $\Law{\Sigma}$.

In Figure~\ref{fig:trans} we introduce a function $\Theta{-}$ that translates sorted formulas $\phi\sort{n}{0}$ into string diagrams of type $n\to 0$. From a general result in~\cite{SeeberThesis}, it follows that $\phi$ is a logical consequence of $\psi$ if and only if $\Theta \psi \leq_{\Sigma, \Pred} \Theta \phi$. For an example take $\psi \equiv \exists x_2. \, P(x_2,x_1) \wedge f(x_1)=x_2 \sort{1}{0}$ and $\phi \equiv \exists x_2. \, P(x_2,x_1)\sort{1}{0}$. Then $\Theta{\psi}$ is the leftmost string diagram below, while $\Theta{\phi}$  is the rightmost one. The following derivation proves that $ \exists x_2. \, P(x_2,x_1)$ is a logical consequence of $\exists x_2. \, P(x_2,x_1) \wedge f(x_1)=x_2$.
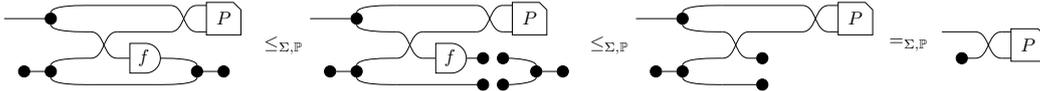
\[
\stikzfig{derivation}
\]
\end{example}

\begin{figure}
\[
\scalebox{0.93}{\ensuremath{
\begin{array}{lclc}
      \Theta { x_i \sort{n}{1} } = \stikzfig{Xi} & (V) &
            \Theta{ f\langle t_1,\dots, t_m\rangle \sort{n}{1}} = \stikzfig{ft1tm}
      & (\Sigma)
\\[1em]
\Theta{\langle \rangle \sort{n}{0}} = \stikzfig{DisN}
      & (\langle \rangle)
&
\Theta{\langle t_1, \dots, t_m\rangle \sort{n}{m}} = \stikzfig{t1Dotstm} & (\langle \dots \rangle)
\\[1.5em]
      \Theta {\top \sort{n}{0}} = \stikzfig{DisN} & (\top)
&
      \Theta { P\langle t_1,\dots, t_{m}\rangle \sort{n}{0}} = \stikzfig{Pt1tm} & (\Pred)
\\[1em]
      \Theta {\exists x_{n+1}\ldotp\phi \sort{n+1}{0}} = {\stikzfig{ExistsPhi}} & (\exists)
&
      \Theta { t_1 = t_2 \sort{n}{0} } = \stikzfig{t1Eqt2} & (=)
\\[1em]
\Theta { \phi \wedge \psi \sort{n}{0}} = {\stikzfig{PhiAndPsi}} & (\wedge)& &
\end{array}
}
}
\]
\caption{Translation $\Theta{-}$ from sorted formulas (Figure \ref{fig:sort}) to string diagrams (Figure \ref{fig:stringdiagrams}).\label{fig:trans}}
\end{figure}

\section{From Cartesian Bicategories to Doctrines}\label{sec: from CBC to EED}
In this section we illustrate how a cartesian bicategory $\B$ gives rise to an elementary existential doctrine.
The starting observation is that, using the conclusion of Lemma~\ref{lemma:homsets in CBC are InfSLattices}, the functor $\Hom_\B(-, I) \from \op{\B} \to \Set$ sends objects $X$ to Inf-semilattices. However, for an arbitrary morphism $R\colon X \to Y$, $\Hom_\B(R, I) \from \Hom_\B(Y,I) \to \Hom_\B(X,I)$ may \emph{not} be an inf-preserving function. A sufficient condition is to require $R$ to be a map; indeed, in that case it is immediate to see that  infima, as defined in Lemma~\ref{lemma:homsets in CBC are InfSLattices}, are preserved.
Thus, by restricting the domain of the $\Hom$-functor to $\op{\Map(\B)}$, one obtains a contravariant functor from the cartesian category $\Map(\B)$ (Lemma~\ref{lem:MapCart}) to $\InfSL$:
\begin{equation}\label{eq:defRB}
\RA(\B) = \Hom_\B(-, I) \from \op{\Map(\B)} \to \InfSL.
\end{equation}

\begin{theorem}\label{thm:doctrine associated to cbc}
The functor $\RA(\B)$ in \eqref{eq:defRB} is an elementary existential doctrine.
\end{theorem}
\begin{proof}
	We need to show that $\RA(\B)$ is elementary and existential.
\begin{enumerate}\item First we prove that $\RA(\B)$ is elementary. We fix  $\delta_A = \stikzfig{DeltaA}\in \Hom(A\otimes A, I)$.
The string diagram for $e = id_X \times \Delta_A
			 \from X \times A \to X \times A \times A$ is $\raisebox{3pt}{\scalebox{0.50}{\stikzfig{IdCopy}}}$.
			The function $\RA(\B)_e$ maps $\stikzfig{TripleR}$ to $\stikzfig{ETripleR}$.
The function $\exists_e \from \RA(\B)(A \times X) \to \RA(\B)(A \times X \times X)$ defined in \eqref{eq:elementary2} maps every $R\in \Hom(X \otimes A, I)$ to
			\[
			\exists_e(R) = \RA(\B)_{\langle \pi_1, \pi_2 \rangle}(R) \wedge \RA(\B)_{\langle \pi_2, \pi_3 \rangle}(\delta_A) = \stikzfig{ExistsE}
			\]
Now $\exists_{e}$ is left-adjoint to $\RA(\B)_{e}$ because $\stikzfig{Copy}$ is left-adjoint to $\stikzfig{Cocopy}$. \item To show that $\RA(\B)$ is existential, let $\pi \from X \times A \to A$ be a projection. Then $\exists_{\pi}$ maps
		$\stikzfig{DoubleR2}$ to $\stikzfig{ExistsR}$. This is left-adjoint to $\RA(\B)_{\pi}$ since $\stikzfig{Codis}$ is left-adjoint to $\stikzfig{Dis}$.
		\begin{itemize}\item For the Beck-Chevalley condition, consider the diagram in Remark~\ref{rem:Beck-Chevalley pullback diagram simple}.
We need to show that
			$
			\exists_{\pi_2}(\RA(\B)_{\id_X \times f}(\beta)) = \RA(\B)_{f}(\exists_{\pi}(\beta))
			$.
Translated to diagrams,
			\[\stikzfig{Beck}\]
			which holds trivially.
			\item For Frobenius reciprocity, take
a projection $\pi \from X \times A \to A$, $\alpha \in \RA(\B)(A)$ and $\beta \in \RA(\B)(X \times A)$.
			We need that
			$\exists_{\pi}(\RA(\B)_{\pi}(\alpha) \wedge \beta) = \alpha \wedge \exists_{\pi}(\beta)
			$, which translates to
			\[
			\stikzfig{FrobRec}
			\]
			This holds by naturality of the symmetry and since $\stikzfig{Dis}$ is the counit of $\stikzfig{Copy}$. \qedhere
\end{itemize}
	\end{enumerate}
\end{proof}

By applying $\RA$ to
the cartesian bicategory $\Rel$, one obtains $\pow \from \op{\Set} \to \InfSL$, in the sense that $ \pow \cong \RA(\Rel)$ in $\EED$. The isomorphism is the pair $(\Gamma,b)$, where $\Gamma \from \Set  \to \Map(\Rel)$ is the strict cartesian functor computing the graph of a function and $b \colon \pow \to \RA(\Rel) \circ \op\Gamma$ is the natural transformation defined for all sets $A$ and $S\in \pow(A)$ as $b_A(S)=\{ (s,\bullet) \mid s \in S \}$.

\begin{example}\label{ex:R(CB)=LT}

Consider $\fCB{\Sigma}{P}$ of Example \ref{ex:CBSP}: then $\RA(\fCB{\Sigma}{P})$ is isomorphic to the doctrine $LT\colon \op{\Law{\Sigma}} \to \InfSL$ of Example \ref{example:Lindenbaum-Tarski eed}. Here is why: the first two rows in Figure \ref{fig:trans} define a cartesian isomorphism $F_{\Theta{\cdot}} \colon \Law{\Sigma} \to \Map(\fCB{\Sigma}{\Pred})$. For all $n\in \Nat$, we define $b_n \colon LT(n) \to \Hom(F_{\Theta{\cdot}}(n),0)=\Hom(n,0)$ as the function mapping $\phi \sort{n}{0}$ to $\Theta{\phi} \colon n \to 0$. This give rise to a natural isomorphism $b\colon LT \to \RA ( \fCB{\Sigma}{\Pred})\circ \op{F_{\Theta{\cdot}}}$.
To show that $b$ preserves equalities and existential quantifiers (Definition \ref{def:category EED}) it suffices to note that  $\Theta{x_1=x_2} = \stikzfig{Cap}$ and $\Theta {\exists x_{n+1}\ldotp\phi} = \raisebox{0pt}{\stikzfig{ExistsPhi}}$. The pair $(F_{\Theta{-}},b)$ witnesses the isomorphism of $LT$ and $\RA ( \fCB{\Sigma}{\Pred})$.
\end{example}

\begin{proposition}
Assigning doctrines to cartesian bicategories as in~\eqref{eq:defRB} extends to a functor $\RA \colon \CBC \to \EED$. For a morphism of cartesian bicategories $F \colon \A \to \B$ in $\CBC$,
$
\RA(F)=({F}\restriction_{\Map(\A)}, b^F)
$ 
and $b^F \colon \RA(\A) \to \RA(\B)\circ  \op{(F\restriction_{\Map(\A)})}$ is defined as	
\begin{equation}\label{eq:functor R transformation component}
	b^F_A(U) = F(U) \in \Hom_{\B} (F(A),I) \quad \text{for all $A \in \A$ and $U \in \Hom_\A (A,I)$}.
	\end{equation}
\end{proposition}
\begin{proof}
	Let $F \colon \A \to \B$ be a morphism in $\CBC$. By Proposition~\ref{prop:restriction of a morphism in CBC to maps}, we have that $F\restriction_{\Map(\A)}$ is a strict cartesian functor.
Regarding $b^F$, its naturality is ensured by (in fact, equivalent to) the functoriality of $F$, while the preservation of equalities and existential quantifiers of $\RA(\A)$ follows from the fact that $F$ preserves the structure of cartesian bicategory of $\A$, which is used to define the structure of elementary existential doctrine of $\RA(\A)$. Therefore $\RA(F)$ is indeed a morphism in $\EED$.
Preservation of compositions and identities is straightforward.
\end{proof}
 \section{From Doctrines to Cartesian Bicategories}\label{sec: from EED to CBC}

Given an elementary existential doctrine $P \colon \op\C \to \InfSL$, we can form a category $\Bicat_P$ whose objects are the same as $\C$ and whose morphisms are given by the elements of $P(X \times Y)$, intuitively seen as \emph{relations}, as observed in~\cite{maietti2015unifying}. Inspired by the calculus of ordinary relations on sets, using the structure of $P$ we can define a notion of composition and tensor product of these inner relations, and endow each object with a comonoid structure that makes $\Bicat_P$ a cartesian bicategory. Here we recall the essential definitions, while the proof of the fact that $\Bicat_P$ actually satisfies Definition~\ref{def:cartBicat} is rather laborious and therefore omitted here: the interested reader can find it in all details in Appendix~\ref{appendix: EED to CBC}.

Since $\Hom_{\Bicat_{P}}(X,Y) = P(X \times Y)$, we get that $\Bicat_P$ is poset-enriched. Given that $P$ is elementary, the obvious candidate for identity on $X$ is
$\delta_X \in \Hom_{\Bicat_P}(X,X)$. Composition works as follows:
let $f \in \Hom_{\Bicat_P}(X,Y) = P(X \times Y)$ and $g \in \Hom_{\Bicat_P}(Y,Z) = P(Y \times Z)$.
\[
\begin{tikzcd}
& X \times Y \times Z \arrow[swap]{dl}{\pi_Z} \arrow{d}{\pi_Y} \arrow{dr}{\pi_X} & \\
X \times Y & X \times Z & Y \times Z
\end{tikzcd}
\]
Consider the projections above. Then the composite $f \seq g$ is defined as
\[
f \seq g = \exists_{\pi_Y}(P_{\pi_Z}(f) \wedge P_{\pi_X}(g)).
\]
The monoidal structure of $\Bicat_P$ is very straightforward: on objects, the monoidal product $\otimes$ is given by the cartesian product in $\Cat$.
On morphisms, for $f \in \Hom_{\Bicat_P}(A,B) = P(A \times B)$ and $g \in \Hom_{\Bicat_P}(C,D) = P(C \times D)$,
consider
the projections
\[
\begin{tikzcd}
& A \times C \times B \times D \arrow[swap]{dl}{\langle \pi_1, \pi_3 \rangle} \arrow{dr}{\langle \pi_2, \pi_4 \rangle} & \\
A \times B & & C \times D
\end{tikzcd}
\]
Then
let
\[
f \otimes g = P_{\langle \pi_1, \pi_3 \rangle}(f) \wedge P_{\langle \pi_2, \pi_4 \rangle}(g) .
\]This makes $\Bicat_P$ a monoidal, poset-enriched category. The rest of the structure of cartesian bicategory is inherited from $\C$ by means of a crucial tool: the \emph{graph functor} of $P$, hereafter denoted by  $\Gamma_P \colon \C \to \Bicat_P$. It is the identity on objects and sends arrows $f \colon X \to Y$ in $\C$ to
\[
\Gamma_P(f) = P_{f \times \id_Y}(\delta_Y) \in P(X \times Y) = \Hom_{\Bicat_P}(X,Y).
\]

\begin{proposition}\label{prop:GammaP(f) has right adjoint}
Let $P \colon \op\C \to \InfSL$ be an elementary, existential doctrine. Then:
\begin{itemize}
    \item $\Gamma_P$ is strict monoidal,
    \item $\Gamma_P(f)$ has a right adjoint, namely $P_{\id_Y \times f} (\delta_Y)$, for every $f \colon X \to Y$ in $\C$. Therefore, by Lemma~\ref{lemma:MapRightAdj}, it corestricts to a strict cartesian functor $\Gamma_P \colon \C \to \Map(\Bicat_P)$.
\end{itemize}
\end{proposition}

Consider for instance the powerset doctrine $\pow \from \op{\Set} \to \InfSL$: the elements of $\pow(X \times Y)$ are precisely the relations from $X$ to $Y$, while composition and monoidal product in $\Bicat_\pow$ coincide with the usual composition and product of relations, see \S~\ref{sec:cartesianbi}. In other words, $\Bicat_{\pow}=\Rel$. The functor $\Gamma_{\pow}$ calculates graphs of functions, which are exactly the maps in $\Rel$.

\begin{example}
Recall the doctrine $LT \colon \Law{\Sigma} \to \InfSL$  and the cartesian bicategory $\fCB{\Sigma}{\Pred}$ from Examples \ref{example:Lindenbaum-Tarski eed} and \ref{ex:CBSP}. The functor $\Gamma_{LT}\colon \Law{\Sigma} \to \Map(\Bicat_{LT})$ is inductively defined as the unique cartesian functor mapping each $f \in \Sigma$ with arity $\ar(f)=n$ to the formula $f(x_1, \dots, x_n)=x_{n+1} \sort{n+1}{0}$.
\end{example}

Now, since $\C$ is a cartesian category, every object is canonically equipped with a natural and coherent comonoid structure. We can use the strict monoidal functor $\Gamma_P \colon \C \to \Bicat_P$ to transport this comonoid to $\Bicat_P$, and by Proposition~\ref{prop:GammaP(f) has right adjoint} we have that copying and discarding in $\Bicat_P$ both have right adjoints. 
Finally, one can prove that every morphism in $\Bicat_P$ is a lax-comonoid homomorphism and that the Frobenius law holds.

\begin{theorem}\label{thm: LA is a functor}
Let $P \colon \op\C \to \InfSL$ be an elementary existential doctrine. Then $\Bicat_P$ is a cartesian bicategory. Moreover, the assignment $P \mapsto \Bicat_P$ extends to a functor $\LA \colon \EED \to \CBC$ as follows:
for $P$ as above, $R \colon \op\D \to \InfSL$ and $(F,b) \colon P \to R$ in $\EED$, 
	\[
	\begin{tikzcd}[row sep=0em]
	\Bicat_P \ar[r,"{\LA(F,b)}"] & \Bicat_R \\
	X \ar[r,|->] \ar[d,"P(X \times Y) \ni r"'] \ar[d,draw=none,""name=x] & FX \ar[d,"b_{X\times Y}(r) \in R(F(X) \times F(Y))"] \ar[d,draw=none,swap,""name=y]\\[2em]
	Y \ar[r,|->] & FY
	\arrow[|->,from=x,to=y]
	\end{tikzcd}
	\]
\end{theorem}

 \section{An Adjunction}\label{sec:adjunction}

We saw in Sections~\ref{sec: from CBC to EED} and~\ref{sec: from EED to CBC} that  using $\LA$ and $\RA$ one can  pass, in a functorial way, between the worlds of cartesian bicategories and elementary existential doctrines. Here we show that, in fact, they define an adjunction $\LA \dashv \RA$. For this we need natural transformations
\[
\eta \colon \id_{\EED} \to \RA\LA, \qquad \epsilon \colon \LA\RA \to \id_\CBC
\]
that make the following triangles commute for every $P \colon \op\C \to \InfSL$ in $\EED$ and $\B$ in $\CBC$.
\begin{equation}\label{eq:triangular equalities diagrams}
\begin{tikzcd}
		\LA(P) \ar[r,"\LA(\eta_P)"] \ar[dr,"\id_{\LA(P)}"'] & \LA\RA\LA(P) \ar[d,"\epsilon_{\LA(P)}"] \\
		& \LA(P)
		\end{tikzcd}
\qquad
	\begin{tikzcd}
		\RA(\B) \ar[r,"\eta_{\RA(\B)}"] \ar[dr,"\id_{\RA(\B)}"'] & \RA\LA\RA(\B) \ar[d,"\RA(\epsilon_\B)"] \\
		& \RA(\B)
		\end{tikzcd}
\end{equation}

Let us start with $\epsilon$. Recall that $\RA(\B) = \Hom_{\B}(-,I) \colon \op{\Map(\B)} \to \InfSL$ and that, for $f \colon X \to Y$ in $\Map(\B)$ and $U \in \Hom_\B(Y,I)$,
$\RA(\B)_f(U)$ is equal to $U \circ f \colon X \to I$. By definition then, $\LA\RA(\B)=\Bicat_{\RA(\B)}$ has the same objects of $\B$ while a morphism $X \to Y$ in $\LA\RA(\B)$ is an element of $\Hom_\B(X \times Y, I)$. Hence, $\epsilon_\B$ has to take a $\sxtikzfig{RfromXYtoI}$ in $\B$ and produce a morphism $\epsilon_\B(R) \colon X \to Y$ in $\B$. We can do so by ``bending'' the $Y$ string, defining:
    \begin{equation}\label{eq:epsilon}
    \epsilon_\B\left(\stikzfig{RfromXYtoI} \right) \quad = \quad \stikzfig{epsilonR}
    \end{equation}
    In fact, by the snake equation in~\eqref{eq:snake}, $\epsilon_\B$ has an inverse: define, for $S \colon X \to Y$ in $\B$,
    \[
    \inv\epsilon_\B \left(\stikzfig{S}\right)\quad  = \quad \stikzfig{epsilonINVs}
    \]
    and it is immediate to see that $\inv\epsilon_\B \epsilon_\B(R)=R$ and $\epsilon_\B\inv\epsilon_\B(S)=S$.

\begin{example}
Recall from Example~\ref{ex:R(CB)=LT} that $\RA(\fCB{\Sigma}{P}) \cong LT$. Applying $\LA$ to both, one gets $\LA\RA(\fCB{\Sigma}{P}) \cong \LA(LT)=\Bicat_{LT}$, hence using $\epsilon$ we have that $\fCB{\Sigma}{P} \cong \Bicat_{LT}$. In particular, given a formula $\phi \sort{n+m}{0}$, one obtains a morphism $n \to m$ in $\fCB{\Sigma}{P}$ by first translating $\phi$ to the string diagram $ \Theta{\phi}\colon n+m \to 0$ (which is a morphism in $\LA\RA(\fCB{\Sigma}{P})$ of type $n \to m$), and then bending the last $m$ inputs using $\epsilon$ as illustrated in \eqref{eq:epsilon}.
\end{example}

\begin{remark}\label{rem:inv epsilon = Gamma R(B) on maps}
$\inv\epsilon_\B$ coincides with $\Gamma_{\RA(\B)}$ on $\Map(\B)$, since $\inv\epsilon_\B (f) = \delta^{\RA(\B)}_Y \circ (f \times \id_Y) = \Gamma_{\RA(\B)}(f)$ when $f \colon X \to Y$ is a map. In other words, $\inv\epsilon_\B$ is an extension of $\Gamma_P$ to the whole of $\B$.
\end{remark}

Regarding $\eta$, we have $\LA(P)=\Bicat_P$, whose objects are the objects of $\C$, and hom-sets are $\Bicat_P(X,Y)=P(X \times Y)$. This means that
	\[
	\RA\LA(P) = \Hom_{\Bicat_P}(-,I) = P({-}\times I) \colon \op{\Map({\Bicat_P})} \to \InfSL.
	\]
	To give a morphism $\eta_P \colon P \to \RA\LA(P)$ in $\EED$ means therefore to give a functor $F \colon \C \to \Map(\Bicat_P)$ and a natural transformation $b \colon P \to P(F(-) \times I)$ satisfying certain conditions. We have a natural candidate for $F$: we proved in Proposition~\ref{prop:GammaP(f) has right adjoint} that $\Gamma_P$ is a functor whose image, in fact, is included in $\Map(\Bicat_P)$ and, moreover, it is cartesian. Being the identity on objects, the natural transformation part of the definition of $\eta_P$ must have components $b_X \colon P(X) \to P(X \times I)$: it is clear then that the natural transformation
	\[
	P\rho= \bigl( P_{\rho_X} \colon P(X) \to P(X \times I) \bigr)_{X \in \C},
	\]
    obtained by whiskering $P$ with the right unitor $\rho_X \colon X \times I \to X$ of $\C$, is a sensible choice for $b$. In short, we define
    \begin{equation}\label{eq:eta}
    \eta=\bigl( (\Gamma_P, P\rho) \colon P \to \RA\LA(P) \bigr)_{P \in \EED}.
    \end{equation}

    The interested reader can find a proof of the fact that $\eta$ and $\epsilon$ are well-defined natural transformations that satisfy the triangular equalities~\eqref{eq:triangular equalities diagrams} in Appendix~\ref{appendix:adjunction}. 

    \begin{theorem}\label{thm:adjunctionNEW}
	The functors $\LA$ and $\RA$ form an adjunction
	\begin{equation}\label{eq:adjunction}
	\begin{tikzcd}[column sep=5em]
		\EED \ar[r,bend left,"\LA"] \ar[r,draw=none,"\bot"description] & \CBC \ar[l,bend left,"\RA"]
	\end{tikzcd}
	\end{equation}
	whose unit is $\eta$~\eqref{eq:eta} and whose counit, which is a natural isomorphism, is $\epsilon$~\eqref{eq:epsilon}.
\end{theorem}

Since the counit $\epsilon$ of the adjunction $\LA \dashv \RA$ is actually a natural isomorphism, the functor $\RA \colon \CBC \to \EED$ is full and faithful.
It turns out, however, that the adjunction $\LA \dashv \RA$ is not an equivalence, because $\LA$ is not faithful. The following example shows why.

\begin{example}\label{ex:counter}
Let $\Sigma_1$ and $\Sigma_2$ be the signatures in Example~\ref{ex:Lawmorphism} and $\Pred$ be some signature of predicate symbols.
Let $LT_1 \colon \op{\Law{\Sigma_1}} \to \InfSL$  be the indexed Lindenbaum-Tarski algebras for $\Sigma_1$ defined as in Example~\ref{example:Lindenbaum-Tarski eed}.
Recall from Example~\ref{ex:Lawmorphism} the strict cartesian functor $Q \colon \Law{\Sigma_2} \to \Law{\Sigma_1}$, mapping $g_1,g_2\in \Sigma_2$ to $f\in \Sigma_1$, and observe that \[LT_1' = LT_1 \circ \op Q \colon \op{\Law{\Sigma_2}} \to \InfSL\]
is an elementary, existential doctrine by Remark~\ref{rem:precomposition of doctrine with cartesian functor}.
Its behaviour is somewhat peculiar: it maps $n\in \Nat$ to the set of formulas $\phi \sort{n}{0}$ built from $\Sigma_1$ and $\Pred$ and to any tuple of terms $\langle t_1, \dots, t_m \rangle \colon n \to m$ in $\Sigma_2$ assigns the function mapping a formula $\phi \sort{m}{0}$ to $\phi [^{Q(t_1), \dots, Q(t_m)}   / _{x_1, \dots, x_m}]$. Observe that each $Q(t_i)$ is again a term in $\Sigma_1$: the symbols $g_1$ and $g_2$ are both translated to $f$. Somehow $LT_1'$ behaves like $LT_1$ but they are different doctrines since their index categories, $\Law{\Sigma_2}$ and $\Law{\Sigma_1}$, are not isomorphic. However, when transforming them into cartesian bicategories via $\LA$, one obtains that $\LA(LT_1')=\LA(LT_1)$: objects are natural numbers and morphisms $n \to m$ are the elements of the set
\[
LT_1'(n + m) = LT_1(Q(n + m)) = LT_1(Q(n) + Q(m)) = LT_1(n+m).
\]
To formally show that $\LA$ is not faithful we now define two morphisms in $\EED$, both from $LT_1$ to $LT_1'$, and show that their image along $\LA$ is the same. Consider the strict cartesian functors $F_1,F_2 \colon \Law{\Sigma_1} \to \Law{\Sigma_2}$ from Example~\ref{ex:Lawmorphism} and observe that
$QF_i=\id_{\Law{\Sigma_1}}$ for $i=1,2$. We are in the following situation:
\[
\raisebox{3pt}{
\scalebox{0.85}{
\begin{tikzcd}[bend angle=10,ampersand replacement=\&]
{\op{\Law{\Sigma_1}}} \ar[dr,"LT_1"] \ar[dd,bend left,xshift=.75em,"\op{F_2}"] \ar[dd,bend right,"\op{F_1}"',xshift=-.75em]  \ar[loop,looseness=3.5,"\id"'] \\
\& \InfSL  \\
\op{\Law{\Sigma_2}} \ar[uu,"\op Q"description] \ar[ur,"LT_1'"']
\end{tikzcd}
}
}
\quad \text{ which means that } LT_1' \circ \op{F_i} = LT_1 \circ \op Q \circ \op{F_i} = LT_1,
\]
thus $(F_1,\id_{LT_1})$ and $(F_2,\id_{LT_1})$ are distinct morphisms in $\EED$ from $LT_1$ to $LT_1'$. Since $\LA(LT_1)=\LA(LT_1')$, it follows from the definition of $\LA$ that
$
\LA(F_1,\id) = \id_{\Bicat_{LT}} = \LA(F_2,\id).
$
\end{example}

 \section{An Equivalence}\label{sec:equivalence}
In the previous section we identified a doctrine with peculiar behaviour, and used it to show that~\eqref{eq:adjunction} is not an equivalence. Here we characterise the additional constraints on doctrines that are needed for an equivalence.

To make the adjunction~\eqref{eq:adjunction} an equivalence, we need its unit $\eta \colon \id_\EED \to \RA\LA$ to be a natural isomorphism. This would mean that
\[
\eta_P = (\Gamma_P,P\rho) \colon P \to \RA\LA(P)
\]
ought to be invertible in $\EED$ for any elementary existential doctrine $P \colon \op\C \to \InfSL$. Since
\[
\RA\LA(P)=\Hom_{\Bicat_P} (-,I) \colon \op{\Map(\Bicat_P)} \to \InfSL,
\]
by definition $\eta_P$ has an inverse in $\EED$ if and only if the functor $\Gamma_P \colon \C \to \Map(\Bicat_P)$ is an isomorphism of cartesian categories, in which case $(\inv\Gamma_P, P\inv\rho)$ would be the inverse of $\eta_P$.

But $\Gamma_P$ is the identity on objects: to be an isomorphism, full and faithful suffices. That is, the maps of $\Bicat_P$ must bijectively correspond
with the arrows of the base category $\C$ of $P$. This is not necessarily the case, as arrows of $\C$ are not involved in the construction of $\Bicat_P$.

The technical tools that allow us to bridge the gap between morphisms in $\C$ and maps in $\Bicat_P$ are provided by~\cite{maietti_triposes_2017}: the next two definitions are equivalent to faithfulness and, respectively, fullness of $\Gamma_P$.

\begin{definition} 
Let $P \colon \op\C \to \InfSL$ be an elementary existential doctrine and $\alpha \in P(A)$. A \emph{comprehension} of $\alpha$ is an arrow $\{ \! | \alpha |\!\} \colon X \to A$ in $\C$ such that $\top_{PX} \le P_{\{ \! | \alpha |\!\}} (\alpha)$ and such that for every $h \colon Z \to A$ for which $\top_{PZ} \le P_h (\alpha)$, there is a unique arrow $h' \colon Z \to X$ such that $h= \{\!| \alpha | \! \} \circ h'$. $P$ \emph{has comprehensive diagonals} if every diagonal arrow $\Delta_A \colon A \to A \times A$ is the comprehension of $\delta^P_A$.
\end{definition}
\begin{example}
The doctrine $LT_1'\colon \op{\Law{\Sigma_2}} \to \InfSL$ from Example \ref{ex:counter} does not have comprehensive diagonal. Indeed $\Delta_1$ is not the comprehension of $\delta_1$: take as $h$ the arrow $\langle g_1,g_2 \rangle\colon 1 \to 2$ and observe that ${LT_1'}_{\langle g_1,g_2 \rangle}(\delta_1)=\bigl(f(x_1)=f(x_1)\bigr) \sort{1}{0}$ and $\top \sort{1}{0} \leq \bigl(f(x_1)=f(x_1)\bigr) \sort{1}{0}$. Yet there exists no $h'\colon 1 \to 1$ in $\Law{\Sigma_2}$ such that $\langle g_1,g_2 \rangle = \Delta_1 \circ h'$.
\end{example}
\begin{definition}
Let $P \colon \op\C \to \InfSL$ be an elementary, existential doctrine. We say that $P$ \emph{satisfies the Rule of Unique Choice} (RUC) if for every  $R \in P(X \times Y)$ which is a map in $\Bicat_P$ there exists an arrow $f \colon X\to Y$ such that
$
\top_{PX} \le P_{\langle \id_X, f \rangle} (R).
$
\end{definition}
\begin{example}
All the doctrines considered so far satisfy RUC. For an example of a doctrine that does not satisfy it consider the composition of $\op{F_1} \colon \op{\Law{\Sigma_1}} \to \op{\Law{\Sigma_2}}$ (Example \ref{ex:Lawmorphism}) and $LT_2 \colon \op{\Law{\Sigma_2}} \to \InfSL$ (Example~\ref{example:Lindenbaum-Tarski eed}) that, by Remark \ref{rem:precomposition of doctrine with cartesian functor}, is an elementary existential doctrine. This doctrine maps $n$ to the set of formulas $\phi\sort{n}{0}$ where terms are built from $g_1$ and $g_2$, but the index category $L_{\Sigma_1}$ contains terms built from $f$ that is translated to $g_1$ by $F_1$. Now, the formula $\phi \equiv \bigl( g_2(x_1)=x_2 \bigr) \sort{2}{0}$ belongs to $LT_2 \circ \op{F_1}(1+1)$ and gives rise to a
map in $\Bicat_{LT_2 \circ \op{F_1}}$, but there is no arrow $t \colon 1 \to 1$ in $\Law{\Sigma_1}$ such that $\top \le {LT_2}_{\langle \id, F_1(t) \rangle} (\phi)$.
\end{example}
We denote by $\EEDff$ the full sub-category of $\EED$ consisting only of those elementary existential doctrines with comprehensive diagonals and satisfying the Rule of Unique Choice. Conveniently, it turns out that the image of $\RA$ is already contained in $\EEDff$.

\begin{proposition}
Let $\B$ be a cartesian bicategory. Then $\RA(\B)$ is in $\EEDff$.
\end{proposition}
\begin{proof}It is enough to show that $\Gamma_{\RA(\B)} \colon \Map(\B) \to \Map(\Bicat_{\RA(\B)}) = \Map(\LA\RA(\B))$ is full and faithful. As noticed in Remark~\ref{rem:inv epsilon = Gamma R(B) on maps},
$
\Gamma_{\RA(\B)} = \inv\epsilon_\B \restriction_{\Map(\B)}. $
Since $\inv\epsilon_\B$ is an isomorphism in $\CBC$, it is faithful, therefore its restriction $\Gamma_{\RA(\B)}$ is as well. Moreover, $\inv\epsilon_\B$ is full: if $R \colon X \to Y$ in $\LA\RA(\B)$ is a map, then there exists $f \colon X \to Y$ in $\B$ such that $\inv\epsilon_\B (f) = R$. In fact, $f=\epsilon_\B(R)$ and since $\epsilon_\B$ is a morphism in $\CBC$, by Proposition~\ref{prop:restriction of a morphism in CBC to maps} we have that $f$ is a map in $\B$. Therefore, $\Gamma_{\RA(\B)}$ is full.
\end{proof}

\begin{theorem}\label{thm:EEDff and CBC are equivalent}
The categories $\EEDff$ and $\CBC$ are equivalent via adjunction~\eqref{eq:adjunction} where $\LA$ and $\RA$ are respectively restricted and corestricted to $\EEDff$.
\end{theorem}

\section{Conclusion}\label{sec:conclusion}

We gave an exhaustive analysis of the relationship between two different categorifications of regular logic: the \emph{universal} approach of elementary existential doctrines, and the \emph{algebraic} approach of cartesian bicategories. We showed that cartesian bicategories give rise to elementary existential doctrines and, expanding a remark in \cite{maietti2015unifying}, that also the other direction is possible. We proved that this correspondence is functorial, in the sense that we have a pair of functors $\LA \colon \EED \to \CBC$ and $\RA \colon \CBC \to \EED$ which are moreover adjoint (Theorem \ref{thm:adjunctionNEW}).

This adjunction can be strengthened to an equivalence provided that we refine the notion of doctrine, excluding some problematic examples (e.g.\ Example \ref{ex:counter}). These cases lay outside the image of $\RA$ and thus this restriction does not affect cartesian bicategories (Theorem \ref{thm:EEDff and CBC are equivalent}).

We hope that understanding the relationship between $\CBC$ and $\EED$ may provide some hints on the nature of the additional algebraic structure needed for cartesian bicategories to capture full first order logic, which one can do on the $\EED$, universal side by considering Lawvere's original hyperdoctrines. It is probable that the end result will be closely related to Peirce's existential graphs~\cite{peirce2019_the-logic-of-the-future}, a 19th century proto-string-diagrammatic logical syntax. This direction has already started to be explored, from diverse perspectives, in~\cite{Brady2000,Brady2000a,Haydon2020}.

\bibliography{Bibliography}

\begin{thebibliography}{10}

\bibitem{BaezErbele-CategoriesInControl}
John Baez and Jason Erbele.
\newblock Categories in control.
\newblock {\em Theory and Application of Categories}, 30(24):836--881, 2015.
\newblock URL: \url{http://www.tac.mta.ca/tac/volumes/30/24/30-24abs.html}.

\bibitem{DBLP:journals/pacmpl/BonchiHPSZ19}
Filippo Bonchi, Joshua Holland, Robin Piedeleu, Pawel Sobocinski, and Fabio
  Zanasi.
\newblock Diagrammatic algebra: from linear to concurrent systems.
\newblock {\em Proc. {ACM} Program. Lang.}, 3({POPL}):25:1--25:28, 2019.
\newblock \href {https://doi.org/10.1145/3290338} {\path{doi:10.1145/3290338}}.

\bibitem{GCQ}
Filippo Bonchi, Jens Seeber, and Pawel Sobocinski.
\newblock {Graphical Conjunctive Queries}.
\newblock In Dan Ghica and Achim Jung, editors, {\em 27th EACSL Annual
  Conference on Computer Science Logic (CSL 2018)}, volume 119 of {\em Leibniz
  International Proceedings in Informatics (LIPIcs)}, pages 13:1--13:23,
  Dagstuhl, Germany, 2018. Schloss Dagstuhl--Leibniz-Zentrum fuer Informatik.
\newblock \href {https://doi.org/10.4230/LIPIcs.CSL.2018.13}
  {\path{doi:10.4230/LIPIcs.CSL.2018.13}}.

\bibitem{Bonchi2015}
Filippo Bonchi, Pawel Sobocinski, and Fabio Zanasi.
\newblock Full {Abstraction} for {Signal} {Flow} {Graphs}.
\newblock In {\em Proceedings of the 42nd {Annual} {ACM} {SIGPLAN}-{SIGACT}
  {Symposium} on {Principles} of {Programming} {Languages}}, {POPL} '15, pages
  515--526, New York, NY, USA, January 2015. Association for Computing
  Machinery.
\newblock \href {https://doi.org/10.1145/2676726.2676993}
  {\path{doi:10.1145/2676726.2676993}}.

\bibitem{DBLP:journals/jlp/BonchiSZ18}
Filippo Bonchi, Pawel Sobocinski, and Fabio Zanasi.
\newblock Deconstructing {L}awvere with distributive laws.
\newblock {\em Journal of Logical and Algebraic Methods in Programming},
  95:128--146, 2018.
\newblock \href {https://doi.org/10.1016/j.jlamp.2017.12.002}
  {\path{doi:10.1016/j.jlamp.2017.12.002}}.

\bibitem{Brady2000a}
Geraldine Brady and Todd Trimble.
\newblock A string diagram calculus for predicate logic and {C}. {S}.
  {P}eirce's system beta.
\newblock Unpublished, available online at
  \url{https://ncatlab.org/nlab/files/BradyTrimbleString.pdf}, 2000.

\bibitem{Brady2000}
Geraldine Brady and Todd~H. Trimble.
\newblock A categorical interpretation of {C}.{S}. {Peirce}'s propositional
  logic {Alpha}.
\newblock {\em Journal of Pure and Applied Algebra}, 149(3):213--239, June
  2000.
\newblock \href {https://doi.org/10.1016/S0022-4049(98)00179-0}
  {\path{doi:10.1016/S0022-4049(98)00179-0}}.

\bibitem{carboni1987cartesian}
Aurelio Carboni and Robert~F.C. Walters.
\newblock {Cartesian Bicategories I}.
\newblock {\em Journal of Pure and Applied Algebra}, 49(1):11--32, 1987.
\newblock \href {https://doi.org/https://doi.org/10.1016/0022-4049(87)90121-6}
  {\path{doi:https://doi.org/10.1016/0022-4049(87)90121-6}}.

\bibitem{coecke2011interacting}
Bob Coecke and Ross Duncan.
\newblock Interacting quantum observables: categorical algebra and
  diagrammatics.
\newblock {\em New Journal of Physics}, 13(4):043016, April 2011.
\newblock Publisher: IOP Publishing.
\newblock \href {https://doi.org/10.1088/1367-2630/13/4/043016}
  {\path{doi:10.1088/1367-2630/13/4/043016}}.

\bibitem{Liberti2021}
Ivan Di~Liberti, Fosco Loregian, Chad Nester, and Paweł Sobociński.
\newblock Functorial semantics for partial theories.
\newblock In {\em Proceedings of the {ACM} on {Programming} {Languages}},
  volume~5, pages 57:1--57:28, January 2021.
\newblock \href {https://doi.org/10.1145/3434338} {\path{doi:10.1145/3434338}}.

\bibitem{Fong2015}
Brendan Fong, Paweł Sobociński, and Paolo Rapisarda.
\newblock A categorical approach to open and interconnected dynamical systems.
\newblock In {\em Proceedings of the 31st {Annual} {ACM}/{IEEE} {Symposium} on
  {Logic} in {Computer} {Science}}, {LICS} '16, pages 495--504, New York, NY,
  USA, July 2016. Association for Computing Machinery.
\newblock \href {https://doi.org/10.1145/2933575.2934556}
  {\path{doi:10.1145/2933575.2934556}}.

\bibitem{DBLP:journals/corr/abs-2009-06836}
Brendan Fong and David Spivak.
\newblock String diagrams for regular logic (extended abstract).
\newblock In John Baez and Bob Coecke, editors, {\em Applied Category Theory
  2019}, volume 323 of {\em Electronic Proceedings in Theoretical Computer
  Science}, page 196–229. Open Publishing Association, Sep 2020.
\newblock \href {https://doi.org/10.4204/eptcs.323.14}
  {\path{doi:10.4204/eptcs.323.14}}.

\bibitem{fong2018graphical}
Brendan Fong and David~I Spivak.
\newblock Graphical regular logic, 2019.
\newblock \href {http://arxiv.org/abs/1812.05765} {\path{arXiv:1812.05765}}.

\bibitem{fox_coalgebras_1976}
Thomas Fox.
\newblock Coalgebras and cartesian categories.
\newblock {\em Communications in Algebra}, 4(7):665--667, 1976.
\newblock \href {https://doi.org/10.1080/00927877608822127}
  {\path{doi:10.1080/00927877608822127}}.

\bibitem{freyd1990categories}
Peter Freyd and Andre Scedrov.
\newblock {\em Categories, {Allegories}}, volume~39 of {\em North-{Holland}
  {Mathematical} {Library}}.
\newblock Elsevier B.V, 1990.
\newblock URL:
  \url{https://www.sciencedirect.com/bookseries/north-holland-mathematical-library/vol/39/suppl/C}.

\bibitem{Ghica2016}
Dan~R. Ghica and Achim Jung.
\newblock Categorical semantics of digital circuits.
\newblock In {\em 2016 {Formal} {Methods} in {Computer}-{Aided} {Design}
  ({FMCAD})}, pages 41--48, October 2016.
\newblock \href {https://doi.org/10.1109/FMCAD.2016.7886659}
  {\path{doi:10.1109/FMCAD.2016.7886659}}.

\bibitem{Haydon2020}
Nathan Haydon and Paweł Sobociński.
\newblock Compositional {Diagrammatic} {First}-{Order} {Logic}.
\newblock In {\em Diagrammatic {Representation} and {Inference}}, Lecture
  {Notes} in {Computer} {Science}, 11th International Conference, Diagrams
  2020, Tallinn, Estonia,, 2020. Springer International Publishing.
\newblock \href {https://doi.org/https://doi.org/10.1007/978-3-030-54249-8_32}
  {\path{doi:https://doi.org/10.1007/978-3-030-54249-8_32}}.

\bibitem{Lack2004a}
Stephen Lack.
\newblock Composing {PROPs}.
\newblock {\em Theory and Application of Categories}, 13(9):147--163, 2004.
\newblock URL: \url{http://www.tac.mta.ca/tac/volumes/13/9/13-09abs.html}.

\bibitem{lawvere1963functorial}
F.~William Lawvere.
\newblock Functorial semantics of algebraic theories.
\newblock {\em Proceedings of the National Academy of Sciences},
  50(5):869--872, 1963.

\bibitem{lawvere_adjointness_1969}
F.~William Lawvere.
\newblock Adjointness in {Foundations}.
\newblock {\em Dialectica}, 23(3/4):281--296, 1969.
\newblock Publisher: Wiley.
\newblock \href {https://doi.org/10.1111/j.1746-8361.1969.tb01194.x}
  {\path{doi:10.1111/j.1746-8361.1969.tb01194.x}}.

\bibitem{lawvere_diagonal_1969}
F.~William Lawvere.
\newblock Diagonal arguments and cartesian closed categories.
\newblock In {\em Category {Theory}, {Homology} {Theory} and their
  {Applications} {II}}, Lecture {Notes} in {Mathematics}, pages 134--145,
  Berlin, Heidelberg, 1969. Springer.
\newblock \href {https://doi.org/10.1007/BFb0080769}
  {\path{doi:10.1007/BFb0080769}}.

\bibitem{lawvere_equality_1970}
F.~William Lawvere.
\newblock Equality in hyperdoctrines and comprehension schema as an adjoint
  functor.
\newblock In Alex Heller, editor, {\em Applications of {Categorical}
  {Algebra}}, volume~17, pages 1--14, New York, NY, 1970. American Mathematical
  Society.
\newblock \href {https://doi.org/https://doi.org/10.1090/pspum/017}
  {\path{doi:https://doi.org/10.1090/pspum/017}}.

\bibitem{MacLane1965}
Saunders MacLane.
\newblock Categorical algebra.
\newblock {\em Bulletin of the American Mathematical Society}, 71(1):40--106,
  1965.
\newblock \href {https://doi.org/10.1090/S0002-9904-1965-11234-4}
  {\path{doi:10.1090/S0002-9904-1965-11234-4}}.

\bibitem{maietti_triposes_2017}
Maria~Emilia Maietti, Fabio Pasquali, and Giuseppe Rosolini.
\newblock Triposes, exact completions, and {Hilbert}'s ε-operator.
\newblock {\em Tbilisi Mathematical Journal}, 10(3):141--166, June 2017.
\newblock Publisher: Tbilisi Centre for Mathematical Sciences.
\newblock \href {https://doi.org/10.1515/tmj-2017-0106}
  {\path{doi:10.1515/tmj-2017-0106}}.

\bibitem{maietti_quotient_2013}
Maria~Emilia Maietti and Giuseppe Rosolini.
\newblock Quotient {Completion} for the {Foundation} of {Constructive}
  {Mathematics}.
\newblock {\em Logica Universalis}, 7(3):371--402, September 2013.
\newblock \href {https://doi.org/10.1007/s11787-013-0080-2}
  {\path{doi:10.1007/s11787-013-0080-2}}.

\bibitem{maietti2015unifying}
Maria~Emilia Maietti and Giuseppe Rosolini.
\newblock Unifying {Exact} {Completions}.
\newblock {\em Applied Categorical Structures}, 23(1):43--52, February 2015.
\newblock \href {https://doi.org/10.1007/s10485-013-9360-5}
  {\path{doi:10.1007/s10485-013-9360-5}}.

\bibitem{DBLP:conf/lics/MuroyaCG18}
Koko Muroya, Steven W.~T. Cheung, and Dan~R. Ghica.
\newblock The geometry of computation-graph abstraction.
\newblock In Anuj Dawar and Erich Gr{\"{a}}del, editors, {\em Proceedings of
  the 33rd Annual {ACM/IEEE} Symposium on Logic in Computer Science, {LICS}
  2018, Oxford, UK, July 09-12, 2018}, pages 749--758. {ACM}, 2018.
\newblock \href {https://doi.org/10.1145/3209108.3209127}
  {\path{doi:10.1145/3209108.3209127}}.

\bibitem{patterson2017knowledge}
Evan Patterson.
\newblock Knowledge representation in bicategories of relations, 2017.
\newblock \href {http://arxiv.org/abs/1706.00526} {\path{arXiv:1706.00526}}.

\bibitem{Piedeleu2021}
Robin Piedeleu and Fabio Zanasi.
\newblock A {String} {Diagrammatic} {Axiomatisation} of {Finite}-{State}
  {Automata}.
\newblock In Stefan Kiefer and Christine Tasson, editors, {\em Foundations of
  {Software} {Science} and {Computation} {Structures}}, Lecture {Notes} in
  {Computer} {Science}, pages 469--489, Cham, 2021. Springer International
  Publishing.
\newblock \href {https://doi.org/10.1007/978-3-030-71995-1_24}
  {\path{doi:10.1007/978-3-030-71995-1_24}}.

\bibitem{peirce2019_the-logic-of-the-future}
Ahti-Veikko Pietarinen.
\newblock {\em The Logic of the Future}, volume~1.
\newblock De Gruyter, 2019.

\bibitem{SeeberThesis}
Jens Seeber.
\newblock {\em Logical completeness for string diagrams}.
\newblock PhD thesis, IMT Lucca, 2020.

\bibitem{selinger2010survey}
P.~Selinger.
\newblock A {Survey} of {Graphical} {Languages} for {Monoidal} {Categories}.
\newblock In B.~Coecke, editor, {\em New {Structures} for {Physics}}, volume
  813 of {\em Lecture {Notes} in {Physics}}, pages 289--355. Springer, Berlin,
  Heidelberg, 2010.
\newblock \href {https://doi.org/10.1007/978-3-642-12821-9_4}
  {\path{doi:10.1007/978-3-642-12821-9_4}}.

\end{thebibliography}

\appendix
\section{Appendix to \S~\ref{sec:eed}}\label{appendix:eed}

In the following, we will write $\Delta_A$ or even $A \to AA$ for the morphism $\langle \id_A, \id_A \rangle$ in a cartesian category. Also, projections $\pi_i \colon A_1 \times \dots \times A_n \to A_i$ or pairings of them $\langle \pi_i, \pi_j \rangle \colon A_1 \times \dots \times A_n \to A_i \times A_j$ will simply be written as $A_1\dots A_n \to A_i$ or $A_1 \dots A_n \to A_i A_j$ respectively. For example, $\pi_2 \colon A \times B \times C \to B$ will be denoted as $ABC \to B$, and $\langle \pi_1,\pi_3 \rangle \colon A \times B \times C \to A\times C$ as $ABC \to AC$. In particular, the symmetry in a cartesian category $\sigma_{A,B}=\langle \pi_2, \pi_1 \rangle \colon A \times B \to B \times A$ will be written as $AB \to BA$.

\paragraph*{Details on the Elementary Existential Doctrine $\pow \colon \op\Set \to \InfSL$} The fact that $\pow \colon \op\Set \to \InfSL$ is an elementary existential doctrine is well-known and the computations to check it are rather easy. Anyway, we think that it might be useful for some readers to report such computations hereafter.

It is convenient to start with~\eqref{eq:elementary1}, even if this property is entailed by Definition \ref{def:eeh} (see Remark \ref{rmk:BC}).
Recall that, for $f \colon Y \to X$ in $\Set$, $\pow_f \colon \pow(X) \to \pow(Y)$ is defined as $\pow_f (Z) = \{ y \in Y \mid f(y) \in Z \}$, in particular for a set $A$,
\[
    \begin{tikzcd}[row sep=0em]
    \pow(A \times A) \ar[r,"\pow_{\Delta_A}"] & \pow(A) \\
    Q \ar[r,|->] & \{ a \in A \mid (a,a) \in Q \}
    \end{tikzcd}
\]
Then the function $\exists_{\Delta_A} \colon \pow(A) \to \pow(A \times A)$ defined as in~\eqref{eq:elementary1}, namely for all $B \subseteq A$
\begin{align*}
    \exists_{\Delta_A}(B) &= \pow_{\pi_1}(B) \land \delta_A\\
                          &= \{ (x,b) \mid b \in B \} \cap \{ (a,a) \mid a \in A \}\\
                          &= \{ (b,b) \mid b \in B \}
\end{align*}
is in fact the left adjoint of $\pow_{\Delta_A}$: indeed for $B \in \pow(A)$ and $Q \in \pow(A \times A)$
\[
\exists_{\Delta_A}(B) \subseteq Q \iff \forall b \in B \ldotp (b,b) \in Q \iff B \subseteq \pow_{\Delta_A}(Q).
\]
Similarly, for $e \colon X \times A \to X \times A \times A$ given by $e(x,a) = (x,a,a)$, the function $\exists_e \colon \pow(X \times A) \to \pow(X \times A \times A)$ defined as in~\eqref{eq:elementary2}:
\begin{align*}
    \exists_e(S) &= \pow_{\langle \pi_1, \pi_2 \rangle} (S) \land \pow_{\langle \pi_2, \pi3 \rangle} (\delta_A) \\
    &= \{ (x,a,b) \mid (x,a) \in S \} \cap \{ (x,a,b) \mid (a,b) \in \delta_A \} \\
    &= \{ (x,a,a) \mid (x,a) \in S \}
\end{align*}
is easily seen to be left adjoint of $\pow_e \colon \pow(X \times A \times A) \to \pow(X \times A)$. Moreover, for $\pi \colon X \times A \to A$, $S \subseteq X \times A$ and $B \subseteq A$,
\begin{align*}
\exists_\pi (S) \subseteq B &\iff \forall a \in A \ldotp \bigl( \exists x \in X \ldotp (x,a) \in S \implies a \in B  \bigr) \\
&\iff \forall (x,a) \in S \ldotp a \in B  \\
&\iff S \subseteq \pow_\pi(B)
\end{align*}
so $\exists_\pi$ is left adjoint to $\pow_\pi$.

Regarding the Beck-Chevalley condition: if $f \colon A' \to A$ and $\pi \colon X \times A \to A$, then the pullback of $f$ against $\pi$ is:
\[
\begin{tikzcd}
X \times A' \ar[r,"\pi'"] \ar[d,"\id_X \times f"'] & A' \ar[d,"f"] \\
X \times A \ar[r,"\pi"] & A
\end{tikzcd}
\]
Take $S \in \pow(X \times A)$. Then $\pow_{\id_X \times f}(S)=\{ (x,a') \mid (x,f(a')) \in S \}$, therefore
\begin{align*}
    \exists_{\pi'} (\pow_{\id_X \times f} (S)) &= \{ a' \in A' \mid \exists x \in X \ldotp (x,a') \in \pow_{\id_X \times f} (S) \} \\
    &= \{ a' \in A' \mid \exists x \in X \ldotp (x,f(a')) \in S \}
\end{align*}
while
\begin{align*}
    \pow_f (\exists_\pi (S)) &= \{ a' \in A' \mid f(a') \in \pow_f (\exists_\pi(S))\} \\
    &= \{ a' \in A' \mid \exists x \in X \ldotp (x,f(a')) \in S \}
\end{align*}
Hence the Beck-Chevalley condition is satisfied.
Finally, for $\pi \colon X \times A \to A$, $B \in \pow(A)$ and $S \in \pow(X \times A)$:
\begin{align*}
    \exists_\pi (\pow_\pi(B) \land S) &= \exists_\pi \bigl( \{ (x,b)  \in X \times A \mid b \in B \} \cap S  \bigr) \\
    &= \exists_\pi \bigl( \{(x,b) \in S \mid b \in B \} \bigr) \\
    &= \{ b \in B \mid \exists x \in X \ldotp (x,b) \in S \} \\
    &= B \cap \{ a \in A \mid \exists x \in X \ldotp (x,a) \in S \} \\
    &= B \land \exists_\pi(S)
\end{align*}
therefore also the Frobenius reciprocity is satisfied.

\paragraph*{Additional General Facts about Elementary Existential Doctrines}

We report a few well known, simple results that are going to be useful later on.

\begin{lemma}\label{lemma:delta_AxB}
	Let $P$ be an existential, elementary doctrine. Then $\delta_A = \exists_{\Delta_A}(\top)$. Moreover,
	\[
	\delta_{A \times B} = P_{ABAB \to AA} (\delta_A) \wedge P_{ABAB \to BB} (\delta_B)
	\]
	for all objects $A$ and $B$.
	\label{lem:DeltaExistsTop}
\end{lemma}
\begin{proof}
	By~\eqref{eq:elementary1}, we have \[\exists_{\Delta}(\top) = P_{\pi_1}(\top) \wedge \delta_A = \top \wedge \delta_A = \delta_A.\]
	As for the second part of the statement:
since the projection $ABAB\to AA$ coincides with the composite $ABAB \to AAB \to AA$ in $\C$, by functoriality of $P$ and the fact that $P_f$ preserves meets for any morphism $f$ in $\C$ we have
	\begin{align*}
	& P_{ABAB \to AA} (\delta_A) \wedge P_{ABAB \to BB} (\delta_B) \\
	&= P_{ABAB \to AABB} \Bigl( P_{AABB \to AAB} \bigl( P_{AAB \to AA} (\delta_A) \bigr) \wedge P_{AABB \to BB} (\delta_B) \Bigr) \\
	&= P_{ABAB \to AABB} \Bigl( \exists_e \bigl( P_{AAB \to AA} (\delta_A) \bigr) \Bigr)
	\end{align*}
	where $e = \id_A \times \id_A \times \Delta_B$, due to~\eqref{eq:elementary2}. Notice that the morphism $\id_A \times \sigma_{B,A} \times \id_B$ is an isomorphism, therefore $P_{ABAB \to AABB}$ is the left adjoint of its inverse. Thus we write $P_{ABAB \to AABB} = \exists_{AABB \to ABAB}$ and we have:
	\begin{align*}
	&\exists_{AABB \to ABAB} \exists_e \bigl( P_{AAB \to AA}(\delta_A) \wedge P_{AAB \to AB} (\top) \bigr) \\
	&= \exists_{AABB \to ABAB} \exists_e \exists_{e'}(\top)
	\end{align*}
	where $e' = \Delta_A \times \id_B$. Since left adjoints compose and ${e'} \seq e \seq (\id_A \times \sigma \times \id_B) = \Delta_{A \times B}$, we obtain
	\begin{align*}
	P_{ABAB \to AA} (\delta_A) \wedge P_{ABAB \to BB} (\delta_B) &= \exists_{\Delta_{A \times B}} (\top) \\
	&= \delta_{A \times B}. \qedhere	
	\end{align*}
\end{proof}

\begin{proposition}
	Let $P \from \op\Cat \to \Poset$ be any functor. Let
	\[
	\begin{tikzcd}
	A \arrow{r}{f} \arrow[swap]{d}{h} & B \arrow{d}{g} \\
	C \arrow{r}{k} & D
	\end{tikzcd}
	\]
	be a commutative diagram such that $P_f$ has a left-adjoint $\exists_f$ and likewise for $k$. Then
	\[
	\exists_f(P_h(\gamma)) \leq P_g(\exists_k(\gamma))
	\]
	for any $\gamma \in P(C)$.
	\label{prop:ExistsSquare}
\end{proposition}
\begin{proof}
	Let $\gamma \in P(C)$. Then $\exists_f(P_h(\gamma)) \leq P_g(\exists_k(\gamma))$ if and only if
	\[
	P_h(\gamma) \leq P_f \bigl(P_g(\exists_k(\gamma))\bigr) = P_h \bigl(P_k(\exists_k(\gamma))\bigr)
	\]
	which is always true because $\gamma \leq P_k (\exists_k(\gamma))$ and $P_h$ is monotone.
\end{proof}

Combining Lemma~\ref{lem:DeltaExistsTop} and Proposition~\ref{prop:ExistsSquare} we obtain the following simple result.

\begin{corollary}\label{cor:P_!(delta_I)=top}
	Let $P \colon \op\Cat \to \InfSL$ be an existential elementary doctrine and $I$ be a terminal object in $\Cat$. Then for all $X$
	\[
	P_{{}! \colon X \to I \times I} (\delta_I) = \top_{P(X)}.
	\]
\end{corollary}

\section{Appendix to \S~\ref{sec: from EED to CBC}}\label{appendix: EED to CBC}

Here we show that the graph functor $\Gamma_P$ of an elementary existential doctrine $P \colon \op\C \to \InfSL$ is indeed functorial (Lemma~\ref{lem:Functor}) and that $\Gamma_P(f)$ has a right adjoint in $\Bicat_P$ for every arrow $f$ in $\C$ (Lemma~\ref{lem:Gamma(f) has a right adjoint}).

\paragraph*{Proof of Proposition~\ref{prop:GammaP(f) has right adjoint}}

\begin{proposition}\label{prop:op}
    The symmetry $\sigma_{A,B} \colon A \times B \to B \times A$ of the cartesian product on $\C$ induces a bijection
    \[
    \begin{tikzcd}[row sep=0em]
    \Hom_{\Bicat_P}(X , Y) \ar[r,"\op\bullet"] & \Hom_{\Bicat_P}(Y,X) \\
    f \ar[r,|->] & P_{\sigma_{Y,X}}(f)
    \end{tikzcd}
    \]
    This is contravariant with respect to composition in $\Bicat_P$, meaning that
	\[
	\op{g} \seq \op{f} = \op{(f \seq g)} \quad \text{and} \quad \op{\delta_X} = \delta_X.
	\]
\end{proposition}
\begin{proof}
	For $f \colon X \to Y$, we define $\op f = P_{\sigma_{Y,X}}(f)$, where $\sigma_{Y,X} \colon Y \times X \to X \times Y$ is the symmetry of $\Cat$. 
	
	Consider also $g \colon Y \to Z$. Then
	\[
	\op{(f \seq g)} = P_{\sigma_{Z,X}} \Bigl( \exists_{XYZ \to XZ} \bigl(P_{XYZ\to XY}(f) \wedge P_{XYZ \to YZ }(g)\bigr)\Bigr)
	\]
	while
	\begin{align*}
	\op g \seq \op f &= \exists_{ZYX \to ZX} \Bigl( P_{ZYX \to ZY} \bigl(P_{\sigma_{Z,Y}}(g)\bigr) \wedge P_{ZYX \to YX} \bigl(P_{\sigma_{Y,X}} (f)\bigr)\Bigr) \\
	&= \exists_{ZYX \to ZX} \Bigl( P_{ZYX \to XYZ} \bigl( P_{XYZ \to YZ}(g) \wedge P_{XYZ \to XY} (f) \bigr)  \Bigr)
	\end{align*}
	which is equal to $\op{(f \seq g)}$ thanks to the Beck-Chevalley condition applied to the pullback:
	\[
	\begin{tikzcd}
	Z \times Y \times X \ar[r,"{\langle \pi_1,\pi_3 \rangle}" ] \ar[d,"{\langle \pi_3, \pi_2, \pi_1 \rangle}"'] & Z \times X \ar[d,"{\sigma_{Z,X}}"] \\
	X \times Y \times Z \ar[r,"{\langle \pi_1,\pi_3 \rangle}"] & X \times Z
	\end{tikzcd} \qedhere
	\]
\end{proof}

\begin{lemma}\label{lem:Functor}
	Let $P \from \op\Cat \to \InfSL$ be an elementary existential doctrine.
	There is a functor $\Gamma_P \from \Cat \to \Bicat_P$, called the \emph{graph functor of $P$}, defined to be the identity on objects and on morphisms defined by sending $f \from X \to Y$ to
	\[\Gamma_P(f) = P_{f \times \id_Y}(\delta_Y) \in P(X \times Y) = \Hom_{\Bicat_P}(X,Y).\]
\end{lemma}
\begin{proof}
Firstly $\Gamma_P$ preserves identities because $\Gamma_P(\id_X) = P_{\id_X \times \id_X}(\delta_X) = \delta_X$. Therefore it remains to show that $\Gamma_P$ preserves
	composition. Let $f \from X \to Y$ and $g \from Y \to Z$ be morphisms in $\Cat$.
	Consider the projections
	\[
	\begin{tikzcd}
	& X \times Y \times Z \arrow[swap]{dl}{\pi_Z} \arrow{d}{\pi_Y} \arrow{dr}{\pi_X} & \\
	X \times Y & X \times Z & Y \times Z
	\end{tikzcd}
	\]
	Then
	\[
	\Gamma_P(f) \seq \Gamma_P(g) = \exists_{\pi_Y}(P_{\pi_Z}(P_{f \times \id_Y}(\delta_Y)) \wedge P_{\pi_X}(P_{g \times \id_Z}(\delta_Z))).
	\]
	Now the diagram
	\[
	\begin{tikzcd}[sep=huge]
	X \times Y \times Z \arrow[swap]{d}{\pi_Z} \arrow{r}{f \times \id_Y \times \id_Z} & Y \times Y \times Z \arrow{d}{\pi} \\
	X \times Y \arrow{r}{f \times \id_Y} & Y \times Y
	\end{tikzcd}
	\]
	commutes and therefore
	\[
	P_{\pi_Z}(P_{f \times \id_Y}(\delta_Y)) = P_{f \times \id_Y \times \id_Z}(P_{\pi}(\delta_Y))
	\]
	and likewise
	\[
	\begin{tikzcd}[sep=huge]
	X \times Y \times Z \arrow[swap]{dr}{\pi_X} \arrow{r}{f \times \id_Y \times \id_Z} & Y \times Y \times Z \arrow{d}{\pi'} \\
	& Y \times Z \\
	\end{tikzcd}
	\]
	commutes, where $\pi'$ projects onto the second and third component.
	Therefore we also have
	\[
	P_{\pi_X}(P_{g \times \id_Z}(\delta_Z))) = P_{f \times \id_Y \times \id_Z}(P_{\pi'}(P_{g \times \id_Z}(\delta_Y))).
	\]
	This gives us
	\[
	\Gamma_P(f) \seq \Gamma_P(g) = \exists_{\pi_Y}( P_{f \times \id_Y \times \id_Z}(P_{\pi}(\delta_Y) \wedge P_{\pi'}(P_{g \times \id_Z}(\delta_Y))).
	\]
	Let $e = \Delta_Y \times \id_Z \from Y \times Z \to Y \times Y \times Z$, then by~\eqref{eq:elementary2}, we have
	\[
	\Gamma_P(f) \seq \Gamma_P(g) = \exists_{\pi_Y}( P_{f \times \id_Y \times \id_Z}(\exists_e(P_{g \times \id_Z}(\delta_Z)))).
	\]
	Using the Beck-Chevalley condition with the square
	\[
	\begin{tikzcd}
	X \times Y \times Z \arrow{r}{\pi_Y} \arrow[swap]{d}{f \times \id_Y \times \id_Z} & X \times Z \arrow{d}{f \times \id_Z} \\
	Y \times Y \times Z \arrow{r}{\rho} & Y \times Z
	\end{tikzcd}
	\]
	where $\rho$ projects onto the first and third component, this becomes
	\[
	\Gamma_P(f) \seq \Gamma_P(g) = P_{f \times \id_Z}(\exists_{\rho}(\exists_e(P_{g \times \id_Z}(\delta_Z))))
	\]
	Now since $e \seq \rho = \id$ and by uniqueness of adjoints, this becomes
	\[
	\Gamma_P(f) \seq \Gamma_P(g) = P_{f \times \id_Z}(P_{g \times \id_Z}(\delta_Z)) = P_{(f \seq g) \times \id_Z}(\delta_Z) = \Gamma_P(f \seq g). \qedhere
	\]
\end{proof}

\begin{remark}\label{rem:F(f times g) = Ff otimes Fg}
	It follow immediately from the definitions of $\Gamma_P$ and of tensor product in $\Bicat_P$ that for any $f,g$ in $\C$, $\Gamma_P(f \times g) = \Gamma_P(g) \otimes \Gamma_P(g)$.
\end{remark}

\begin{lemma}\label{lem:Gamma(f) has a right adjoint}
	Let $P \from \op\Cat \to \InfSL$ be an elementary, existential doctrine, $\Gamma_P \from \Cat \to \Bicat_P$ the functor from Lemma~\ref{lem:Functor}.
	For any morphism $f \from X \to Y$ in $\Cat$, $\op{\Gamma_P(f)}$ is right-adjoint to $\Gamma_P(f)$, that is
	\[
	\delta_X \leq \Gamma_P(f) \seq \op{\Gamma_P(f)} \quad \text{and} \quad \op{\Gamma_P(f)} \seq \Gamma_P(f) \leq \delta_Y.
	\]
\end{lemma}
\begin{proof}
    It is easy to check that $\op{\Gamma_P(f)} = P_{\id \times f}(\delta_Y)$.
	\begin{itemize}
		\item We start by proving $\delta_X \leq \Gamma_P(f) \seq \op{\Gamma_P(f)}$. Let
		\[
		\begin{tikzcd}
		& X \times Y \times X \arrow[swap]{dl}{\pi_{1,2}} \arrow{d}{\pi_{1,3}} \arrow{dr}{\pi_{2,3}} & \\
		X \times Y & X \times X & Y \times X
		\end{tikzcd}
		\]
		be the projections. Then
		\begin{align*}
		\Gamma_P(f) \seq \op{\Gamma_P(f)} &= \exists_{\pi_{1,3}}(P_{\pi_{1,2}}(P_{f \times \id_Y}(\delta_Y)) \wedge P_{\pi_{2,3}}(P_{\id_Y \times f}(\delta_Y))) \\
		&= \exists_{\pi_{1,3}}(P_{f \times \id_Y \times f}(P_{\pi_{1,2}}(\delta_Y) \wedge P_{\pi_{2,3}}(\delta_Y)) \\
		&= \exists_{\pi_{1,3}}(P_{f \times \id_Y \times f}(\exists_e(\delta_Y))
		\end{align*}
		where the last step uses~\eqref{eq:elementary2} with $e = \Delta_Y \times \id_Y$.
		Applying Beck-Chevalley using the pullback		\[
		\begin{tikzcd}
		X \times Y \times X \arrow[swap]{d}{f \times \id \times f} \arrow{r}{\pi_{1,3}} & X \times X \arrow{d}{f \times f} \\
		Y \times Y \times Y \arrow{r}{\pi_{1,3}} & Y \times Y
		\end{tikzcd}
		\]
		yields
		\[ \Gamma_P(f) \seq \op{\Gamma_P(f)} = P_{f \times f}(\exists_{\pi_{1,3}}(\exists_e(\delta_Y))) = P_{f \times f}(\delta_Y)\]
		where $\exists_{\pi_{1,3}} \circ \exists_e = \id$.
		Now applying Lemma~\ref{lem:DeltaExistsTop} and Proposition~\ref{prop:ExistsSquare} yields
		\[ \Gamma_P(f) \seq \op{\Gamma_P(f)} = P_{f \times f}(\exists_{\Delta}(\top)) \geq \exists_{\Delta}(P_f(\top)) = \exists_{\Delta}(\top) = \delta_Y \]
		\item It remains to prove that $\op{\Gamma_P(f)} \seq \Gamma_P(f) \leq \delta_Y$. Let
		\[
		\begin{tikzcd}
		& Y \times X \times Y \arrow[swap]{dl}{\pi_{1,2}} \arrow{d}{\pi_{1,3}} \arrow{dr}{\pi_{2,3}} & \\
		Y \times X & Y \times Y & X \times Y
		\end{tikzcd}
		\]
		be the projections. Then
		\begin{align*}
		\op{\Gamma_P(f)} \seq \Gamma_P(f) &= \exists_{\pi_{1,3}}(P_{\pi_{1,2}}(P_{\id_Y \times f}(\delta_Y)) \wedge P_{\pi_{2,3}}(P_{f \times \id_Y}(\delta_Y))) \\
		&= \exists_{\pi_{1,3}}(P_{\id_Y \times f \times \id_Y}(P_{\pi_{1,2}}(\delta_Y) \wedge P_{\pi_{2,3}}(\delta_Y)) \\
		&= \exists_{\pi_{1,3}}(P_{\id_Y \times f \times \id_Y}(\exists_e(\delta_Y))
		\end{align*}
		where the last step uses~\eqref{eq:elementary2} with $e = \Delta_Y \times \id_Y$.
		Using Proposition~\ref{prop:ExistsSquare} with the commutative square
		\[
		\begin{tikzcd}
		Y \times X \times Y \arrow[swap]{d}{\id \times f \times \id} \arrow{r}{\pi_{1,3}} & Y \times Y \arrow{d}{\id} \\
		Y \times Y \times Y \arrow{r}{\pi_{1,3}} & Y \times Y
		\end{tikzcd}
		\]
		we get
		\[ \op{\Gamma_P(f)} \seq \Gamma_P(f) \leq \exists_{\pi_{1,3}}(\exists_{e}(\delta_Y)) = \delta_Y. \qedhere \]
	\end{itemize}
\end{proof}

\paragraph*{Proof of Theorem~\ref{thm: LA is a functor}}

Here we prove that $\Bicat_P$ is a cartesian bicategory (Theorem~\ref{thm:A_P is a cartesian bicategory}) and that the assignment $P \mapsto \Bicat_P$ extends to a functor $\LA \colon \EED \to \CBC$ (Proposition~\ref{prop: LA is a functor}).

\begin{theorem}	\label{thm:BicatPCat}
	Let $P \from \op\Cat \to \InfSL$ be an existential, elementary doctrine. Then $\Bicat_P$ is a poset-enriched 
category.
\end{theorem}
\begin{proof}
That $\delta_X$ serves as the identity of composition follows easily from~\eqref{eq:elementary2}:
	Let $f \in \Hom_{\Bicat_{P}}(X,Y) = P(X \times Y)$ and
	\[
	\begin{tikzcd}
	& X \times Y \times Y \arrow[swap]{dl}{\pi_{1,2}} \arrow{d}{\pi_{1,3}} \arrow{dr}{\pi_{2,3}} & \\
	X \times Y & X \times Y & Y \times Y
	\end{tikzcd}
	\]
	the projections. Then
	\[
	f \seq \delta_Y = \exists_{\pi_{1,3}}(P_{\pi_{1,2}}(f) \wedge P_{\pi_{2,3}}(\delta_Y)).
	\]
	By~\eqref{eq:elementary2}, we have $P_{\pi_{1,2}}(f) \wedge P_{\pi_{2,3}}(\delta_Y) = \exists_e(f)$ where
	$e = \id_X \times \Delta_Y \from X \times Y \to X \times Y \times Y$.
	Therefore, $f \seq \delta_Y = \exists_{\pi_{1,3}}(\exists_{e}(f))$. Since left-adjoints compose, we have that $\exists_{\pi_{1,3}} \circ \exists_{e}$
	is left-adjoint to $P_{\pi_{1,3} \circ e} = P_{\id} = \id$. By uniqueness of adjoints,  also $\exists_{\pi_{1,3}} \circ \exists_{e}$ is the identity and hence
	\[f \seq \delta_Y = \exists_{\pi_{1,3}}(\exists_{e}(f)) = f.\]
	That also $\delta_X \seq f = f$ follows by Proposition~\ref{prop:op}.

To see that composition is associative, let $f \from A \to B, g \from B \to C, h \from C \to D$. 
We have
	\[
f \seq g = \exists_{ABC \to AC} (P_{ABC \to AB} (f) \land P_{ABC \to BC} (g)).
	\]
	Let us give the inner expression a name, and let $\alpha = P_{ABC \to AB} (f) \land P_{ABC \to BC} (g)$ so that
	\[
f \seq g = \exists_{ABC \to AC} (\alpha).
	\]
	Then
	\begin{align*}
(f \seq g) \seq h &= \exists_{ACD \to AD}(P_{ACD \to AC} (f \seq g) \land P_{ACD \to CD} (h)) \\
    &= \exists_{ACD \to AD} (P_{ACD \to AC}(\exists_{ABC \to AC}(\alpha) \land P_{ACD \to CD} (h))).
	\end{align*}
	Now we can use the Beck-Chevalley condition with the following square of projections:
	\[
\begin{tikzcd}[sep=huge]
	A \times B \times C \times D \arrow{r}{} \arrow{d}{} & A \times C \times D \arrow{d}{} \\
	A \times B \times C \arrow{r}{} & A \times C
	\end{tikzcd}
	\]
	to infer that
	\[ 
P_{{ACD \to AC}}(\exists_{ABC \to AC}(\alpha)) = \exists_{ABCD \to ACD}(P_{ABCD \to ABC}(\alpha))
	\]
	and therefore
	\[
(f \seq g) \seq h = \exists_{ACD \to AD}(\exists_{ABCD \to ACD}(P_{{ABCD \to ABC}}(\alpha)) \wedge P_{{ACD \to CD}}(h)).
	\]
	Using Frobenius reciprocity to get everything into the existential quantifier we get
	\[
(f \seq g) \seq h = \exists_{{ACD \to AD}}(\exists_{{ABCD \to ACD}}(P_{{ABCD \to ABC}}(\alpha) \wedge P_{{ABCD \to CD}}(h)))
	\]
	and therefore by uniqueness of adjoints also
	\[
(f \seq g) \seq h = \exists_{{ABCD \to AD}}(P_{{ABCD \to ABC}}(\alpha) \wedge P_{{ABCD \to CD}}(h)).
	\]
	Now using the definition of $\alpha$ and the fact that any $P_{\pi}$ preserves meet, this expands into
	\[
(f \seq g) \seq h = \exists_{{ABCD \to AD}}(P_{{ABCD \to AB}}(f) \wedge P_{{ABCD \to BC}}(g) \wedge P_{{ABCD \to CD}}(h)).
	\]
	Completely analogously one can prove that also
	\[
f \seq (g \seq h) = \exists_{{ABCD \to AD}}(P_{{ABCD \to AB}}(f) \wedge P_{{ABCD \to BC}}(g) \wedge P_{{ABCD \to CD}}(h))
	\]
	thus composition is therefore associative. 

	Finally, composition is a monotonous operation because $P_\pi$ and $\exists_\pi$ are monotonous functions, thus if $f_1 \le f_2$ then $f_1 ; g \le f_2 ; g$:
	therefore the category $\Bicat_P$ is poset-enriched.\qedhere

\end{proof}

\begin{proposition}\label{prop:A_P is a symmetric monoidal category}
	Let $P \from \op\Cat \to \InfSL$ be an existential, elementary doctrine. Then $\Bicat_P$ is a symmetric monoidal category.
\end{proposition}
\begin{proof}
	We first prove that $\otimes$ is a functor. Let $f_1 \colon A \to B$, $g_1 \colon C \to D$, $f_2 \colon B \to E$, $g_2 \colon D \to F$. We need to show that
	\[
	(f_1 \otimes g_1) \seq (f_2 \otimes g_2) = (f_1 \seq f_2) \otimes (g_1 \seq g_2).
	\]
	Starting from the left-hand side:
	\begin{align*}
	&{\bigl( P_{ACBD \to AB} (f_1) \wedge P_{ACBD \to CD} (g_1) \bigr)} \seq 
	{ \bigl( P_{BDEF \to BE} (f_2) \wedge P_{BDEF \to DF} (g_2) \bigr) } \\
	{}={}&\exists_{ACBDEF \to ACEF}
	\Bigl( 
	P_{ACBDEF \to ACBD} \bigl( P_{ACBD \to AB }(f_1) \wedge P_{ACBD \to CD} (g_1) \bigr) \\
	&\wedge 
	P_{ACBDEF \to BDEF} \bigl( P_{BDEF \to BE }(f_2) \wedge P_{BDEF \to DF} (g_2) \bigr) 
	\Bigr) \\
	{}={}&\exists_{ACBDEF \to ACEF}
	\bigl(
	P_{ACBDEF \to AB}(f_1) \wedge P_{ACBDEF \to CD}(g_1)\\
	&\wedge P_{ACBDEF \to BE}(f_2) \wedge P_{ACBDEF}(g_2)
	\bigr) \\
	{}={}&\exists_{ACBDEF \to ACEF}
	\Bigl(
	P_{ACBDEF\to ABE} \bigl( P_{ABE \to AB}(f_1) \wedge P_{ABE \to BE} (f_2) \bigr) \\
	&\wedge P_{ACBDEF \to CDF} \bigl( P_{CDF \to CD}(g_1) \wedge P_{CDF \to DF} (g_2) \bigr)
	\Bigr).
	\end{align*}
	Let us call, for the sake of brevity, 
	\[
	\alpha = P_{ABE \to AB}(f_1) \wedge P_{ABE \to BE} (f_2) \text{ and } \beta = P_{CDF \to CD}(g_1) \wedge P_{CDF \to DF} (g_2).
	\]
	Then we proved that
	\[
	(f_1 \otimes g_1) \seq (f_2 \otimes g_2) = \exists_{ACBDEF \to ACEF} \bigl( P_{ACBDEF \to ABE}(\alpha) \wedge P_{ACBDEF \to CDF} (\beta) \bigr).
	\]
	We notice that the projection $ACBDEF \to ACEF$, over which we compute the existential quantifier, can be performed in two steps: first we forget $D$, then $B$. Similarly, the projection $ACBDEF \to ABE$, whose image along $P$ we calculate in $\alpha$, can be decomposed in two steps. We use this to put ourselves in the position of applying the Frobenius reciprocity property of $P$:
	\begin{align*}
	& \exists_{ACBDEF \to ACEF} \bigl( P_{ACBDEF \to ABE}(\alpha) \wedge P_{ACBDEF \to CDF} (\beta) \bigr) \\
	=&\exists_{ACBEF \to ACEF} \\
	&\exists_{ACBDEF \to ACBEF} \Bigl( P_{ACBDEF \to ACBEF} \bigl( P_{ACBEF \to ABE} (\alpha) \bigr) \\
	& \hspace{9em}\wedge P_{ACBDEF \to CDF} (\beta) \Bigr) \\
	=&\exists_{ACBEF \to ACEF} \Bigl( P_{ACBEF \to ABE}(\alpha) \\
	& \hspace{7.5em}\wedge \exists_{ACBDEF \to ACBEF} \bigl( P_{ACBDEF \to CDF} (\beta) \bigr) \Bigr) .
	\end{align*}
	Now we use the Beck-Chevalley condition on the following pullback, consisting only of projections:
	\[
	\begin{tikzcd}
	ACBDEF \ar[r] \ar[d] & ACBEF \ar[d] \\
	CDF \ar[r] & CF
	\end{tikzcd}
	\]
	resulting in exchanging the order of $\exists$ and $P$, obtaining:
	\[
	\exists_{ACBEF \to ACEF} \Bigl( P_{ACBEF \to CF} \bigl( \exists_{CDF \to CF} (\beta) \bigr) \wedge P_{ACBEF \to ABE} (\alpha) \Bigr).
	\]
	By writing the projection $ACBEF \to CF$ as a the composite
	\[
	ACBEF \to ACEF \to CF,
	\] 
	one can again use the Frobenius reciprocity and then the Beck-Chevalley condition to finally obtain
	\[
	P_{ACEF \to CF} \bigl( \exists_{CDF \to CF} (\beta) \bigr) \wedge P_{ACEF \to AE} \bigl( \exists_{ABE \to AE} (\alpha) \bigr)
	\]
	which is equal, by definition, to $(f_1 \seq f_2) \otimes (g_1 \seq g_2)$, as required. Thus $\otimes$ preserves compositions; preservation of identities, which is tantamount to
	\[
	P_{ABAB \to AA} (\delta_A) \wedge P_{ABAB \to BB} (\delta_B) = \delta_{A \times B},
	\]
	follows from Lemma~\ref{lem:DeltaExistsTop}. Therefore $\otimes$ is a functor.
	
	Next, we show that $\otimes$ is coherently associative, unitary and commutative. 
We have already noticed in Remark~\ref{rem:F(f times g) = Ff otimes Fg} that the functor $\Gamma_P$ from Lemma~\ref{lem:Functor} transports the product functor of $\Cat$ into the tensor functor $\otimes$ of $\Bicat_P$: this suggests to transfer the whole cartesian structure of $\Cat$ into $\Bicat_P$, thus the associator, the unitors and the symmetry for $\otimes$ are defined as the image of their correspondent isomorphisms in $\Cat$. This automatically ensures that they are all natural isomorphisms (being whiskerings of a functor with natural isomorphisms) coherent with $\otimes$.
\end{proof}

\begin{theorem}\label{thm:A_P is a cartesian bicategory}
	Let $P \from \op\Cat \to \InfSL$ be an elementary existential doctrine. Then $\Bicat_P$ is a cartesian bicategory.
\end{theorem}
\begin{proof}
	By Theorem~\ref{thm:BicatPCat} and Proposition~\ref{prop:A_P is a symmetric monoidal category}, $\Bicat_P$ is a poset-enriched, symmetric monoidal category.
Since $\Cat$ is a cartesian category, every object is canonically equipped with a comonoid given by copying and discarding.
	We can use the functor $\Gamma_P \from \Cat \to \Bicat_P$ to transport
	this comonoid to $\Bicat_P$ and Lemma~\ref{lem:Gamma(f) has a right adjoint} shows that copying and discarding both have right-adjoints.
	Since the comonoids in $\Cat$ are coherent with the monoidal structure, the same is true in $\Bicat_P$ as the structure is transported through the
	strict monoidal functor $\Gamma_P$.
	It therefore remains to prove that every morphism is a lax-comonoid homomorphism and that the Frobenius law holds.
To see that every morphism is a lax-comonoid homomorphism, we need to prove that
	\begin{equation}\label{eq:every morphism is a lax comonoid homo in A_P}
	\stikzfig{LaxCopy} \quad \text{and} \quad  \stikzfig{LaxDis}
	\end{equation}
	for any morphism $R \in \Hom_{\Bicat_P}(X,Y) = P(X \times Y)$.
	The latter is very easy to see because the right-hand side is the top element in $P(X \times I)$, see Corollary~\ref{cor:P_!(delta_I)=top}. So we focus on the former inequality.
	
	We start with the left-hand side. 
In order to distinguish among the three occurrences of $Y$ in the composite of $R$ with the diagonal in $\Bicat_P$, we will write $Y_1$, $Y_2$ and $Y_3$, even though they all represent the same object $Y$ of $\Bicat_P$. We have:
	\begin{align*}
	&\exists_{X Y_1 Y_2 Y_3 \to X Y_2 Y_3} \Bigl( P_{X Y_1 Y_2 Y_3 \to X Y_1} (R) \wedge P_{X Y_1 Y_2 Y_3 \to Y_1 Y_2 Y_3} \bigl(P_{\Delta_Y Y_2 Y_3} (\delta_{Y^2}) \bigr) \Bigr) 
	\end{align*}
	and since $\delta_{Y^2} = P_{\langle \pi_1, \pi_3 \rangle} (\delta_Y) \wedge P_{\langle \pi_2, \pi_4 \rangle} (\delta_Y)$ by Lemma~\ref{lem:DeltaExistsTop}, we get
	\begin{align*}
	&\exists_{X Y_1 Y_2 Y_3 \to X Y_2 Y_3} \bigl( P_{X Y_1 Y_2 Y_3 \to X Y_1}(R) \wedge P_{X Y_1 Y_2 Y_3 \to Y_1 Y_2} (\delta_Y) \wedge P_{X Y_1 Y_2 Y_3 \to Y_2 Y_3} (\delta_Y)  \bigr) \\
	&=\exists_{X Y_1 Y_2 Y_3 \to X Y_2 Y_3} \Bigl( P_{X Y_1 Y_2 Y_3 \to X Y_1 Y_2} \bigl( P_{X Y_1 Y_2 \to X Y_1} (R) \wedge P_{X Y_1 Y_2 \to Y_1 Y_2} (\delta_Y) \bigr) \\
	&\hspace{9em} \wedge P_{X Y_1 Y_2 Y_3 \to Y_1 Y_3} (\delta_Y) \Bigr).
	\end{align*}
	Using~\eqref{eq:elementary2}, the symmetry in $\Cat$ and~\eqref{eq:elementary2} again, we obtain:
	\begin{align*}
	&\exists_{X Y_1 Y_2 Y_3 \to X Y_2 Y_3}
	\Bigl( 
	P_{X Y_1 Y_2 Y_3 \to X Y_1 Y_2} \bigl( \exists_{X\Delta_Y} (R) \bigr) 
	\wedge
	P_{X Y_1 Y_2 Y_3 \to Y_1 Y_3} (\delta_Y)
	\Bigr) \\
	&=\exists_{X Y_1 Y_2 Y_3 \to X Y_2 Y_3}
	\Bigl[
	P_{X\sigma Y} \Bigl(
	P_{X Y_2 Y_1 Y_3 \to X Y_2 Y_1}
	\bigl( P_{X \sigma} ( \exists_{X \Delta_Y} (R)) \bigr)  \\
	& \hspace{12em} \wedge P_{X Y_2 Y_1 Y_3 \to Y_1 Y_3} (\delta_Y)
	\Bigr) \Bigr] \\
	&= \exists_{X Y_1 Y_2 Y_3 \to X Y_2 Y_3}
	\Bigl[
	P_{X\sigma Y} \Bigl(
	\exists_{XY\Delta_Y} \bigl(
	P_{X \sigma} (\exists_{X \Delta_Y} (R))
	\bigr)
	\Bigr)
	\Bigr].
	\end{align*}
	Since $P_{X \sigma Y}$ and $P_{X \sigma}$ are isomorphisms, they are left adjoints (to their inverses): by composing all the left adjoints in the above expression we obtain
	\[
	\exists_{X \times \Delta_Y} (R).
	\] 
	The right-hand side, instead, is
	\begin{align*}
	&\exists_{X X^2 Y^2 \to X Y^2} 
	\Bigl(
	P_{X X^2 Y^2 \to X X^2} \bigl( P_{\Delta_X \times X^2} (\delta_{X^2}) \bigr)
	\wedge
	P_{X X^2 Y^2 \to X^2 Y^2} (R \otimes R)
	\Bigr) \\
	&=\exists_{X X^2 Y^2 \to X Y^2} 
	\Bigl(
	P_{\Delta_X \times X^2 \times Y^2} 
	\bigl( P_{X^2 X^2 Y^2 \to X^2 X^2} (\delta_{X^2}) \wedge P_{X^2 X^2 Y^2 \to X^2 Y^2} (R \otimes R) \bigr)
	\Bigr) \\
	&=\exists_{X X^2 Y^2 \to X Y^2} 
	\Bigl(
	P_{\Delta_X \times X^2 \times Y^2} 
	\bigl( \exists_{\Delta_{X^2} \times Y^2} (R \otimes R) \bigr).
	\Bigr)
	\end{align*}
	Using the Beck-Chevalley condition on the pullback
	\[
	\begin{tikzcd}
	X X^2 Y^2 \ar[r,"{\langle \pi_1, \pi_3 \rangle}"] \ar[d,"\Delta_X \times X^2 \times Y^2"'] & X Y^2 \ar[d,"{\Delta_X\times Y^2}"] \\
	X^2 X^2 Y^2 \ar[r,"{\langle \pi_1, \pi_3 \rangle}"] & X^2 Y^2
	\end{tikzcd}
	\]
	we get
	\begin{align*}
	&P_{\Delta_X \times Y^2} \bigl( \exists_{\langle \pi_1, \pi_3 \rangle} \exists_{\Delta_{X^2 \times Y^2}} (R \otimes R) \bigr) \\
	&= P_{\Delta_X \times Y^2} (R \otimes R).
	\end{align*}

	Using that \[R \otimes R = P_{\langle\pi_1,\pi_3\rangle}(R) \wedge P_{\langle\pi_2, \pi_4\rangle}(R) \]
	we have
\[
	P_{\Delta_X \times Y^2} (R \otimes R)  = P_{X \times \pi_1 \colon X \times Y^2 \to X\times Y}(R) \wedge P_{X \times \pi_2 \colon X \times Y^2 \to X\times Y}(R) 
	\]
	Since we have established the right-hand side as a meet, it suffices to prove that \[\exists_{X \times \Delta_Y}(R) \leq P_{X \times \pi_i}(R)\]
	for $i \in \{1,2\}$ and since
	\[ P_{X \times \Delta_Y}(P_{X \times \pi_i}(R)) = R \]
	this follows by virtue of $\exists_{X \times \Delta_Y}$ being a left-adjoint.

	Lastly, we need to check the Frobenius law. For that it suffices to prove that
	\[ \stikzfig{FrobSimple} \]
	which are elements of $P(A^4)$. 
In the following we will calculate in detail the left-hand side of the equation above, where there are seven occurrences of the same object $A$. For notational convenience, we want to distinguish them, like this:
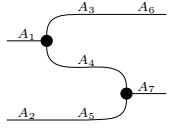
\[
	\stikzfig{FrobSimpleAlessio}
	\] 
	To name a projection out of a product of copies of $A$, we will use the subscripts to specify which factors we are projecting onto: for example, $\langle \pi_1, \pi_2, \pi_3 \rangle \colon A^7 \to A_1 \times A_2 \times A_3$ will simply be denoted as $1234567 \to 123$.
	
	By simply unravelling the definition of composition and tensor, and using the fact that $\delta_{A \times A} = \delta_A \otimes \delta_A$ (Lemma~\ref{lem:DeltaExistsTop}), it is easy to see that the left-hand side is equal to:
	\begin{align*}
	\exists_{1234567\to1267} \bigl(
	&P_{1234567\to13}(\delta_A) \wedge
	P_{1234567\to14}(\delta_A) \wedge
	P_{1234567\to47}(\delta_A) \\
	&\wedge 
	P_{1234567\to57}(\delta_A) \wedge
	P_{1234567\to25}(\delta_A) \wedge
	P_{1234567\to36}(\delta_A) 
	\bigr)
	\end{align*}
	which can be rewritten, using~\eqref{eq:elementary2}, as
	\begin{align*}
	\exists_{1234567\to1267} \Bigl(
	&P_{1234567\to136} \bigl(\exists_{A \times \Delta_A} (\delta_A) \bigr) \wedge
	P_{1234567\to147} \bigl(\exists_{A \times \Delta_A} (\delta_A) \bigr) \\
	&\wedge
	P_{1234567\to257} \bigl(\exists_{A \times \Delta_A} (\delta_A) \bigr)
	\Bigr).
	\end{align*}
	The projection $1234567\to1267$ of the existential quantifier can be decomposed in three projections which forget only one instance of $A$ at a time. Using this and rearranging the expression above, we can put ourselves in the position of being able to use the Frobenius reciprocity as follows: the above is equal to
	\begin{align*}
	&\exists_{12567\to1267} \exists_{124567\to12567} \\
	&\exists_{1234567\to124567}
	\Bigl[
	P_{1234567\to124567} \Bigl(
	P_{124567\to147} \bigl(\exists_{A \times \Delta_A}(\delta_A) \bigr)
	\\
	&\hspace{14.5em}\wedge
	P_{124567\to257} \bigl(\exists_{A \times \Delta_A}(\delta_A) \bigr)
	\Bigr) \\
	&\hspace{7em}\wedge 
	P_{1234567\to136} \bigl(\exists_{A \times \Delta_A}(\delta_A) \bigr)
	\Bigr] \\
	=& \exists_{12567\to1267}\\
	&\exists_{124567\to12567} 
	\Bigl[
	P_{124567\to147} \bigl(\exists_{A \times \Delta_A}(\delta_A) \bigr) 
	\wedge
	P_{124567\to257} \bigl(\exists_{A \times \Delta_A}(\delta_A) \bigr)
	\\
	&\hspace{6em}\wedge
	\exists_{1234567\to124567} \Bigl( P_{1234567\to136} \bigl(\exists_{A \times \Delta_A}(\delta_A) \bigr)  \Bigr)
	\Bigr]
	\end{align*}
	By the Beck-Chevalley condition applied to the following pullback, made just of projections:
	\[
	\begin{tikzcd}
	1234567 \ar[r] \ar[d] & 124567 \ar[d] \\
	136 \ar[r] & 16
	\end{tikzcd}
	\]
	we have that
	\begin{align*}
	\exists_{1234567\to124567} \Bigl( P_{1234567\to136} \bigl(\exists_{A \times \Delta_A}(\delta_A) \Bigr) &=
	P_{124567 \to 16} \Bigl( \exists_{136\to16} \bigl(\exists_{A \times \Delta_A}(\delta_A)\bigr)\Bigr) \\
	&= P_{124567\to16}(\delta_A).
	\end{align*}
	Using this fact and a similar strategy as before, we can reduce the left-hand side to
	\[
	P_{1267\to16}(\delta_A) \wedge P_{1267\to17}(\delta_A) \wedge P_{1267 \to 27} (\delta_A).
	\]
	We can use symmetries in $\Cat$ to arrange the above expression so that we can use~\eqref{eq:elementary2} in the following:
	\begin{align*}
	&P_{1267\to16}(\delta_A) \wedge P_{1267\to17}(\delta_A) \wedge P_{1267 \to 27} (\delta_A)	\\
	&=P_{1267\to2716}\Bigl[
	P_{2716\to271} \Bigl(P_{271\to27}(\delta_A) \wedge P_{271\to71} \bigl(P_\sigma(\delta_A) \bigr) \Bigr) \wedge P_{2716\to16}(\delta_A)
	\Bigr] \\
	&=P_{1267\to2716} \Bigl[
	\exists_{A^2 \times \Delta_A} \Bigl(
	\exists_{A \times \Delta_A} \bigl(
	\exists_{\Delta_A} (\top)
	\bigr)
	\Bigr)
	\Bigr]
	\end{align*}
	Now, $P_{1267\to2716}$ is the left-adjoint of its inverse, which composed to the other existential quantifiers yields simply
	\[
	\exists_{\langle \id_A, \id_A, \id_A, \id_A \rangle} (\top).
	\]
Analogously, one can prove that the right-hand side of the Frobenius equation corresponds to the same expression and therefore
	\[ \stikzfig{FrobSimple} \qedhere \]
\end{proof}

\begin{proposition}\label{prop: LA is a functor}
The assignment $P \mapsto \Bicat_P$ extends to a functor $\LA \colon \EED \to \CBC$ as follows.
\begin{itemize}
	\item On objects: for $P \colon \op\C \to \InfSL$, $\LA(P) = \Bicat_P$.
	\item On morphisms: for $P$ as above, $R \colon \op\D \to \InfSL$ and $(F,b) \colon P \to R$ in $\EED$,
	\[
	\begin{tikzcd}[row sep=0em]
	\Bicat_P \ar[r,"{\LA(F,b)}"] & \Bicat_R \\
	X \ar[r,|->] \ar[d,"P(X \times Y) \ni r"'] \ar[d,draw=none,""name=x] & FX \ar[d,"b_{X\times Y}(r) \in R(F(X) \times F(Y))"] \ar[d,draw=none,swap,""name=y]\\[2em]
	Y \ar[r,|->] & FY
	\arrow[|->,from=x,to=y]
	\end{tikzcd}
	\]
\end{itemize}
\end{proposition}
\begin{proof}

	First we check that $\LA(F,b)$ is a morphism of cartesian bicategories. Let $X$ be an object of $\C$: the identity morphism of $X$ in $\Bicat_P$ is $\delta^P_X$ and
	\[
	\LA(F,b)(\delta^P_X) = b_{X \times X} (\delta^P_X) = \delta^R_{F(X)} = \id^{\Bicat_R}_{F(X)}
	\]
	thus $\LA(F,b)$ preserves identities. Let now $r \in P(X \times Y)$ and $s \in P(Y \times Z)$. Then
	\begin{align*}
		\LA(F,b)(s \underset{\Bicat_P}{\circ} r) &= \LA(F,b) \Bigl( \exists^P_{XYZ \to XZ} \bigl( P_{XYZ \to XY} (r) \wedge P_{XYZ \to YZ}(s) \bigr)  \Bigr) \\
		&= b_{X \times Z} \Bigl( \exists^P_{XYZ \to XZ} \bigl( P_{XYZ \to XY} (r) \wedge P_{XYZ \to YZ}(s) \bigr)  \Bigr) \\
		&= \exists^R_{F(XYZ) \to F(XZ)} \Bigl( b_{X \times Y \times Z} \bigl( P_{XYZ \to XY} (r) \wedge P_{XYZ \to YZ} (s) \bigr) \Bigr) \\
		&= \exists^R_{F(XYZ) \to F(XZ)} \Bigl( b_{X \times Y \times Z} \bigl( P_{XYZ \to XY} (r) \bigr) \wedge b_{X \times Y \times Z} \bigl( P_{XYZ \to YZ}(s) \bigr) \! \Bigr) \\
		&= \exists^R_{F(XYZ) \to F(XZ)} \Bigl( R_{F (XYZ \to XY)} (b_{X \times Y} (r) \bigr) \wedge R_{F (XYZ \to YZ)} \bigl( b_{Y \times Z} (s) \bigr) \! \Bigr) \\
		&= b_{X \times Y} (s) \underset{\Bicat_R}{\circ} b_{X \times Y} (r) \\
		&=\LA(F,b) (s) \underset{\Bicat_R}{\circ} \LA(F,b)(r)
	\end{align*}
	where we used the fact that $b$ preserves existential quantifiers (third equation) and infima (fourth equation) and that $b$ is natural (fifth equation). Therefore $\LA(F,b)$ is a functor, we need to check that it preserves all the structure of cartesian bicategory of $\Bicat_P$.
	
	Preservation of the monoidal product of $\Bicat_P$ is done in an analogous way to the proof above of preservation of compositions, where commutativity of $b$ with existential quantifiers is not needed. Regarding the symmetry, recall that $\sigma^{\Bicat_P}_{X,Y} = \Gamma_P(\sigma^\C_{X,Y})$, hence:
	\begin{align*}
		\LA(F,b)(\sigma^{\Bicat_P}_{X,Y}) &= b_{X \times Y \times Y \times X} (\sigma^{\Bicat_P}_{X,Y}) \\
		&= b_{X \times Y \times Y \times X} \bigl( P_{\sigma^\C_{X,Y} \times \id_{Y \times X}} (\delta^P_{Y \times X})  \bigr) \\
		&= R_{F (\sigma^\C_{X,Y} \times \id_{Y \times X})} \bigl(b_{Y \times X \times Y \times X} (\delta^P_{Y \times X}) \bigr)
	\end{align*}
	by naturality of $b$. Since $F$ preserves the cartesian structure of $\C$, we have that $F (\sigma^\C_{X,Y} \times \id_{Y \times X}) = \sigma^\D_{FX,FY} \times \id_{FY \times FX}$; moreover, $b_{Y \times X \times Y \times X} (\delta^P_{Y \times X}) = \delta^R_{F(Y \times X)}$, therefore
	\[
	\LA(F,b) (\sigma^{\Bicat_P}_{X,Y}) = R_{\sigma^\D_{FX,FY} \times \id_{FY \times FX}} \bigl( \delta^R_{F(Y \times X)} \bigr) = \sigma^{\Bicat_R}_{FX, FY}.
	\]
	
	Analogously one shows that $\LA(F,b)$ preserves the comonoid structure of any object $X$ in $\Bicat_P$, since it is defined as the image along $\Gamma_P$ of copying and discarding in $\C$.  
	
	Finally, monotonicity of $b$ immediately implies that $\LA(F,b)$ preserves the partial order of $P(X \times Y)$ for every $X,Y$ in $\Bicat_P$. This proves that $\LA(F,b)$ is indeed a morphism of cartesian bicategories.
	
	It is left to prove that $\LA$ is functorial. It is immediate to see that for $P \colon \op\C \to \InfSL$, $\LA(\id_\C,\id_P) = \id_{\Bicat_P}$. As per preservation of compositions in $\EED$: let $R \colon \op\D \to \InfSL$ and $S \colon \op\E \to \InfSL$ be elementary existential doctrines, and let $(F,b) \colon P \to R$ and $(G,c) \colon R \to S$ in $\EED$. We have that $(G,c) \circ (F,b) = (GF, cF \circ b)$, therefore
	\[
	\begin{tikzcd}[row sep=0em,column sep=4em]
		\Bicat_P \ar[r,"{\LA \bigl( (G,c) \circ (F,b) \bigr)}"] & \Bicat_S \\
		X \ar[r,|->] \ar[d,"P(X \times Y) \ni r"'] \ar[d,draw=none,""name=x] & GFX \ar[d,"c_{F(X \times Y)} \bigl(b_{X\times Y}(r) \bigr)"] \ar[d,draw=none,swap,""name=y]\\[2em]
		Y \ar[r,|->] & FY
		\arrow[|->,from=x,to=y]	
	\end{tikzcd}
	\]
	while
	\[
	\begin{tikzcd}[row sep=0em]
	\Bicat_P \ar[r,"{\LA(F,b)}"] & \Bicat_R \ar[r,"{\LA(G,c)}"] & \Bicat_S \\
	X \ar[r,|->] \ar[d,"P(X \times Y) \ni r"'] \ar[d,draw=none,""name=x] & FX \ar[d,"b_{X\times Y}(r)"{name=y,description}]  \ar[r,|->] & GFX \ar[d,"{c_{FX \times FY} \bigl( b_{X \times Y} (r) \bigr)}"] \ar[d,draw=none,swap,""name=z]\\[3em]
	Y \ar[r,|->] & FY \ar[r,|->] & GFY
	\arrow[|->,from=x,to=y]
	\arrow[|->,from=y,to=z]
	\end{tikzcd}
	\]
	The two functors indeed coincide.
\end{proof}

\section{Appendix to \S~\ref{sec:adjunction}}\label{appendix:adjunction}

We first show that $\eta$~\eqref{eq:eta} is a natural transformation in Lemma~\ref{lemma:eta natural transformation} and that $\epsilon$~\eqref{eq:epsilon} is a natural isomorphism in Lemma~\ref{lemma:epsilon natural isomorphism}, and then that they satisfy the triangular equalities required by the definition of adjunction in Theorem~\ref{thm:adjunction in appendix}.

\paragraph*{Proof of Theorem~\ref{thm:adjunctionNEW}}
\begin{lemma}\label{lemma:eta natural transformation}
	There is a natural transformation
	\[
	\eta \colon \id_{\EED} \to \RA \LA
	\]
	whose $P$-th component, for $P \colon \op\C \to \InfSL$ an elementary existential doctrine, is $\eta_P = (\Gamma_P, P\rho)$, where $(P\rho)_X = P_{\rho_X}$, with $\rho_X \colon X \times I \to X$ the right unitor in $\C$ (and $I$ the terminal object).
\end{lemma}
\begin{proof}
	Let us fix $P \colon \op\C \to \InfSL$ an elementary existential doctrine. Then $\LA(P)=\Bicat_P$, whose objects are the objects of $\C$, and hom-sets are $\Bicat_P(X,Y)=P(X \times Y)$. This means that 
	\[
	\RA\LA(P) = \Hom_{\Bicat_P}(-,I) = P({-}\times I) \colon \op{\Map({\Bicat_P})} \to \InfSL.
	\]
	To give a morphism $\eta_P \colon P \to \RA\LA(P)$ in $\EED$ means therefore to give a functor $F \colon \C \to \Map(\Bicat_P)$ and a natural transformation $b \colon P \to P(F(-) \times I)$ satisfying certain conditions. We have proved in Lemma~\ref{lem:Functor}, Remark~\ref{rem:F(f times g) = Ff otimes Fg} and Lemma~\ref{lem:Gamma(f) has a right adjoint} that $\Gamma_P$ is a functor whose image, in fact, is included in $\Map(\Bicat_P)$ and, moreover, it is cartesian. Being the identity on objects, the natural transformation part of the definition of $\eta_P$ must have components $b_X \colon P(X) \to P(X \times I)$, therefore $P\rho$ type-checks. That $P\rho$ is natural is obvious, being the whiskering of a functor with a natural transformation; we need to check that it preserves the equalities and existential quantifiers of $P$.
	
	Regarding equalities, we have to show that $P_{\rho_{A\times A}}(\delta^P_A) = \delta^{\RA\LA(P)}_{A}$. We have that 
	\[
	\delta^{\RA\LA(P)}_A = \stikzfig{DeltaA}
	\]
	which is the composite, in $\Bicat_P$, of $\op{\Gamma_P{(\Delta_A)}}$ and $\Gamma_P(! \colon A \to I)$. Since $\Gamma_P(! \colon A \to I) = \top_{P(A \times I)}$, we have:
	\[
		\delta^{\RA\LA(P)}_A = \exists_{A^2 \times A \times I \to A^2 \times I} \Bigl( P_{A^2 \times A \times I \to A^2 \times A} \bigl( P_{\id_{A^2} \times \Delta_A} (\delta^P_{A \times A}) \bigr) \Bigr).
	\]
	Using the Beck-Chevalley condition of $P$ on the following pullback:
	\[
	\begin{tikzcd}
	A^2 \times A \times I \ar[r,"{\langle \pi_1, \pi_3 \rangle}"] \ar[d,"\rho_{A^2 \times A}"'] & A^2 \times I \ar[d,"\rho_{A^2}"] \\
	A^2 \times A \ar[r,"\pi_1"] & A^2
	\end{tikzcd}
	\]
	and recalling that $\delta^P_{A \times A} = P_{\langle \pi_1,\pi_3 \rangle}(\delta^P_A) \wedge P_{\langle \pi_2, \pi_4 \rangle}(\delta^P_A)$, we obtain:
	\begin{align*}
		\delta^{\RA\LA(P)}_A &= P_{\rho_{A^2}} \Bigl( \exists_{A^2 \times A \to A} \bigl( P_{\id_{A^2} \times \Delta_A} (P_{\langle \pi_1,\pi_3 \rangle}(\delta^P_A) \wedge P_{\langle \pi_2, \pi_4 \rangle}(\delta^P_A)) \bigr) \Bigr) \\
		&= P_{\rho_{A^2}} \Bigl( \exists_{A^2 \times A \to A^2} \bigl( P_{\langle \pi_1, \pi_3 \rangle \colon A^3 \to A^2} (\delta^P_A) \wedge P_{\langle \pi_2, \pi_3 \rangle \colon A^3 \to A^2} (\delta^P_A) \bigr)  \Bigr) \\
		&=P_{\rho_{A^2}} \Bigl( \exists_{A^2 \times A \to A^2} \bigl( \exists_{\id_A \times \Delta_A} (\delta^P_A) \bigr) \Bigr) \\
		&= P_{\rho_{A^2}} (\delta^P_A).
	\end{align*}
	
	Regarding the preservation of the existential quantifiers of $P$, we need to check that the square~\eqref{eq:commutative square b preserves existential quantifier} commutes, which means that, for $r \in P(X \times A)$ and $\pi \colon X \times A \to A$, we have to prove that
	\begin{equation}\label{eq:Prho preserves existential quantifier}
	P_{\rho_A} \bigl( \exists^P_\pi (r) \bigr) = \exists^{\RA\LA(P)}_{\Gamma(\pi)} \bigl( P_{\rho_{X \times A}}(r) \bigr)
	\end{equation}
	(where we will continue to write $\Gamma$ instead of $\Gamma_P$ since there is no possibility of confusion).
	
	First we compute $\exists^{\RA\LA(P)}_{\Gamma(\pi)} \colon \Bicat_P(X \times A, I) \to \Bicat_P(A,I)$: by definition, it maps $u \colon X \times A \to I$ morphism in $\Bicat_P$ to the following composite in $\Bicat_P$:
	\[
	\begin{tikzcd}
	A \ar[r,"\lambda_A^{-1}"] & I \times A \ar[r,"{f}"] & X \times A \ar[r,"u"] & I
	\end{tikzcd}
	\]
	where $f = \stikzfig{CodisX} \times \id_A = \op{\Gamma(\stikzfig{DisX})} \times \delta^P_A$ (product computed in $\Bicat_P$) and $\lambda_A^{-1} = \Gamma(\lambda_A^{-1})$ (the first $\lambda$ being the unitor in $\Bicat_P$, the second in $\C$). We begin by calculating the composition in $\Bicat_P$ of $\Gamma(\lambda_A^{-1})$ with $f$. Since $\op{\Gamma(\stikzfig{DisX})} = P_{\id_I \times !_X}$ and $\id_I \times !_X = !_{I \times X} \colon I \times X \to I \times I$, we have, by Corollary~\ref{cor:P_!(delta_I)=top}, that $\op{\Gamma(\stikzfig{DisX})} = \top_{P(I \times X)}$. Therefore, if we annotate with subscripts different copies of the same object $A$ to indicate the various projections in a concise way:
	\begin{align*}
		f \underset{\Bicat_P}\circ {\lambda_A^{-1}} = \exists^P_{A_1 I A_2 X A_3 \to A_1 A_3} \Bigl( &P_{A_1 I A_2 X A_3 \to A_1 I A_2} \bigl( P_{\lambda_A^{-1} \times \id_{I \times A}} (\delta_{I \times A}) \bigr) \\
		& \wedge P_{A_1 I A_2 X A_3 \to A_2 A_3} (\delta_A) \Bigr).
	\end{align*}
	Since, moreover, 
	\[
	\delta_{I \times A} = P_{I A I A \to II} (\delta_I) \wedge P_{IAIA \to AA} (\delta_A) = \top \wedge P_{IAIA \to AA} (\delta_A) = P_{IAIA \to AA} (\delta_A),
	\]
	we have
	\begin{align*}
		f \underset{\Bicat_P}\circ {\lambda_A^{-1}} &= \exists_{A_1 I A_2 X A_3 \to A_1 X A_3} \bigl( P_{A_1 I A_2 X A_3 \to A_1 A_2}(\delta_A) \wedge P_{A_1 I A_2 X A_3 \to A_2 A_3} (\delta_A)\bigr) \\
		&=\exists_{A_1 I A_2 X A_3 \to A_1 X A_3} \Bigl( P_{A_1 I A_2 X A_3 \to A_1 I A_2 X} \bigl( P_{A_1 I A_2 X \to A_1 A_2} (\delta_A) \bigr) \\
		&\hspace{10em}\wedge P_{A_1 I A_2 X A_3 \to A_2 A_3} (\delta_A) \Bigr) \\
		&=\exists_{A_1 I A_2 X A_3 \to A_1 X A_3} \Bigl( \exists_{A_1 I A_2 X \to A_1 I A_2 X A_2} \bigl(P_{A_1 I A_2 X \to A_1 A_2} (\delta_A) \bigr) \Bigr) \\
		&=\exists_{A_1 I A_2 X \to A_1 X A_2}\bigl(P_{A_1 I A_2 X \to A_1 A_2} (\delta_A) \bigr) \\
		&=P_{A_1 X A_2 \to A_1 I A_2 X} \bigl(P_{A_1 I A_2 X \to A_1 A_2} (\delta_A) \bigr) \\
		&=P_{A_1 X A_2 \to A_1 A_2} (\delta_A).
	\end{align*}
	Composing the above with $P_{\rho_{X \times A}} (r)$, we obtain the right-hand side of~\eqref{eq:Prho preserves existential quantifier}, as follows:
	\begin{align*}
		\exists^{\RA\LA(P)}_{\Gamma(\pi)} \bigl(P_{\rho_{X \times A}} (r) \bigr) &=
		\exists_{AXAI \to AI} \Bigl( P_{AXAI \to AXA} \bigl( P_{AXA \to AA} (\delta_A) \bigr) \\
		&\hspace{6.5em}
		\wedge P_{A_1 X A_2 I \to X A_2 I} \bigl( P_{\rho_{X \times A}} (r) \bigr)\Bigr) \\
		&= \exists_{AXAI \to AI} \bigl( P_{AXAI \to AA} (\delta_A) \wedge P_{A_1 X A_2 I \to XA_2} (r) \bigr) \\
		&=\exists_{AXAI \to AI} \Bigl( P_{A_1 X A_2 I \to X A_2 I} \bigl( P_{X A_2 I \to XA_2} (r) \bigr) \\ &\hspace{6.5em} \wedge P_{AXAI \to AA}(\delta_A) \Bigr) \\
		&=\exists_{AXAI \to AI} \Bigl( \exists_{XAI \to AXAI} \bigl(P_{XAI \to XA}(r)\bigr)\Bigr) \\
		&= \exists_{XAI \to AI} \bigl( P_{\rho_{X \times A}} (r) \bigr).
	\end{align*}
	Therefore equation~\eqref{eq:Prho preserves existential quantifier} reduces to
	\[
	P_{\rho_A} \bigl(\exists^P_{\pi} (r)\bigr) = \exists^P_{XAI \to AI} \bigl( P_{\rho_{X \times A}} (r) \bigr)
	\]
	which holds because the following square is a pullback in $\C$:
	\[
	\begin{tikzcd}
	X \times A \times I \ar[r,"{\langle \pi_2, \pi_3 \rangle}"] \ar[d,"\rho_{X \times A}"'] & A \times I \ar[d,"\rho_A"] \\
	X \times A \ar[r,"\pi"] & A
	\end{tikzcd}
	\]
	therefore we can use the Beck-Chevalley condition of $P$.
	
	Finally, we prove that $\eta$ is a natural transformation. Let $R \colon \op\D \to \InfSL$ and $(F,c) \colon P \to R$ be a morphism in $\EED$. Then
	\[
	\RA\LA(F,c) = \bigl( \LA(F,c)\restriction_{\Map(\Bicat_P)} , b^{\LA(F,c)} \bigr)
	\]
	where recall that 
	\[
	\begin{tikzcd}[row sep=0em]
    \Bicat_P \ar[r,"{\LA(F,c)}"] & \Bicat_R \\
    X \ar[r,|->] \ar[d,"P(X \times Y) \ni r"'] & FX \ar[d,"c_{X \times Y} (r)"] \\[2em]
    Y \ar[r,|->] & FY
	\end{tikzcd}
	\quad \text{and} \quad
	\begin{tikzcd}[row sep=0em]
	P(X \times I) \ar[r,"{b^{\LA(F,c)}}"] & R(FX \times I) \\
	u \ar[r,|->] & c_{X \times I} (u)
	\end{tikzcd}
	\]
	We then need to prove that
	\begin{equation}\label{eq:naturality square eta}
	    \begin{tikzcd}
	    P \ar[r,"\eta_P"] \ar[d,"{(F,c)}"'] & \RA\LA(P) \ar[d,"{\RA\LA(F,c)}"] \\
	    R \ar[r,"\eta_R"'] & \RA\LA(R)
	    \end{tikzcd}
	\end{equation}
	commutes in $\EED$. In the first component, we are asking for the following square of functors to commute:
	\[
	\begin{tikzcd}
	\C \ar[r,"\Gamma_P"] \ar[d,"F"'] & \Map(\Bicat_P) \ar[d,"{\LA(F,c)\restriction_{\Map(\Bicat_P)}}"] \\
	\D \ar[r,"\Gamma_R"'] & \Map(\Bicat_R)
	\end{tikzcd}
	\]
	The upper and lower legs send a morphism $f \colon X \to Y$ in $\C$ to
	\[
	c_{X \times Y} \bigl( P_{f \times \id_Y} (\delta^P_Y) \bigr) \quad \text{and} \quad R_{Ff \times \id_{FY}} (\delta^R_{FY})
	\]
	respectively, and those expressions are the same map in $\Bicat_R$ due to the naturality condition of $c \colon P \to R\op F$ applied to the morphism $f \times \id_Y \colon X \times Y \to Y \times Y$, which states that
	\[
	\begin{tikzcd}
	P(X \times Y) \ar[r,"c_{X \times Y}"] & RF(X \times Y) \\
	P(Y \times Y) \ar[u,"P_{f \times \id_Y}"] \ar[r,"c_{Y \times Y}"'] & RF(Y \times Y) \ar[u,"RF_{f \times \id_Y} = R_{Ff \times \id_{FY}}"']
	\end{tikzcd}
	\]
	commutes, hence, given that $\delta^P_Y \in P(Y \times Y)$, we have
	\[
	c_{X \times Y} \bigl( P_{f \times \id_Y} (\delta^P_Y) \bigr) = R_{Ff \times \id_{FY}} \bigl(c_{Y \times Y} (\delta_Y) \bigr) = R_{Ff \times \id_{FY}} (\delta^R_{FY}).
	\]
	In the second component, \eqref{eq:naturality square eta} becomes
	\[
	\begin{tikzcd}
	P \ar[r,"P\rho"] \ar[d,"c"'] & P(- \times I) \ar[d,"c_{- \times I}"] \\
	R \ar[r,"R\rho"'] & R(- \times I)
	\end{tikzcd}
	\]
	which, on the $X$-th component, is a square in $\InfSL$:
	\[
	\begin{tikzcd}
	P(X) \ar[r,"P(\rho_X)"] \ar[d,"c_X"'] & P(X \times I) \ar[d,"c_{X \times I}"] \\
	RF(X) \ar[r,"R(\rho_{FX})"'] & R(FX \times I)
	\end{tikzcd}
	\]
	(Notice that $\rho_{FX} = F(\rho_X)$ because $F$ is strict cartesian.) This square does indeed commute due to the naturality of $c$ applied to the morphism $\rho_X \colon X \to  X \times I$ in $\op \C$.
	\end{proof}
	
	It is convenient to analyse the category $\LA\RA(\B)$ for a cartesian bicategory $\B$. Recall that $\RA(\B) = \Hom_{\B}(-,I)$ and that, for $f \colon X \to Y$ in $\Map(\B)$ and $u \colon Y \to I$ morphism in $\B$, $\RA(\B)_f(u)= u \circ f \colon X \to I$. By definition then, $\LA\RA(\B)=\Bicat_{\RA(\B)}$ has the same objects of $\B$ while a morphism $X \to Y$ in $\LA\RA(\B)$ is an element of $\B(X \times Y, I)$. 

\begin{itemize}
	\item \textbf{Composition.} 
If $R \colon X \times Y \to I$ and $S \colon Y \times Z \to I$ are morphisms of $\B$, then
\begin{equation}\label{eq:composition in LARA(B)}
\begin{aligned}
	S \underset{\LA\RA(\B)}{\circ} R &= \exists_{XYZ \to XZ}^{\RA(\B)} \bigl( \RA(\B)_{XYZ \to XY} (R) \wedge \RA(\B)_{XYZ \to YZ} (S) \bigr) \\
	&=\exists_{XYZ \to XZ}^{\RA(\B)} \Bigl( \stikzfig{RApiR} \wedge \stikzfig{RApiS} \Bigr) \\
	&=\exists_{XYZ \to XZ}^{\RA(\B)} \left( \stikzfig{RApiRinfRApiS} \right) \\
	&= {\stikzfig{ScircRinLARAB}}
\end{aligned}
\end{equation}
(where the above projections appearing in the subscripts are projections in $\Map(\B)$, which are given by ${\stikzfig{Dis}}$ and the unitors).

\item \textbf{Monoidal product.} Take $R \colon X \times X' \to I$ and $S \colon Y \times Y' \to I$ in $\B$, which means that $R \colon X \to X'$ and $S \colon Y \to Y'$ in $\LA\RA(\B)$. Then
\begin{equation}\label{eq:tensor in LARA(B)}
\begin{aligned}
	R \otimes S &= \RA(\B)_{XYX'Y' \to XX'} (R) \land \RA(\B)_{XYX'Y' \to YY'} (S) \\
				&= {\stikzfig{RApiF}} \quad \wedge \quad {\stikzfig{RApiG}} \\
				&= {\stikzfig{RApiFinfRApiG}} \\
				&= {\stikzfig{fTensorGinLARA}}
\end{aligned}
\end{equation}
\end{itemize}

\begin{lemma}\label{lemma:epsilon natural isomorphism}
	There is a natural isomorphism
	\[
	\epsilon \colon \LA\RA \to \id_{\CBC}
	\]
	whose $\B$-th component, for $\B$ a cartesian bicategory, sends $R \in \LA\RA(\B)(X,Y)=\B(X\times Y,I)$ to the composite
	\[
	\begin{tikzcd}
	X \ar[r,"\rho_X^{-1}"] & X \times I \ar[r,"\bar R"] & I \times Y \ar[r,"\lambda_Y"] & Y
	\end{tikzcd}
	\]
	where $\rho$ and $\lambda$ are, respectively, the right and left unitors in $\B$ and $\bar R$ is:
	\begin{equation}\label{eq:epsilon(r)}
	\stikzfig{EpsilonR}
	\end{equation}
	The inverse $\epsilon_\B^{-1}$ of $\epsilon_\B$ sends $S \colon X \to Y$ in $\B$ to
	\[
	\stikzfig{epsilonINVs}
	\]
\end{lemma}
\begin{proof}
	The plan of the proof is as follows: first we show that $\epsilon_\B$ and $\epsilon_\B^{-1}$ are inverse of each other as functions on hom-sets, in the sense that for any $R \colon X \times Y \to I$ and $S \colon X \to Y$ in $\B$, $\inv\epsilon_\B \epsilon_\B (R)=R$ and $\epsilon_\B \inv\epsilon_\B (S)=S$. Next we observe that $\inv\epsilon_\B$ and $\Gamma_{\RA(B)}$ coincide on $\Map(\B)$. This immediately implies that $\inv\epsilon_\B$ preserves symmetries and the comonoid structure on every object of $\B$, while the Snake Lemma will imply that the same holds for $\epsilon_\B$. Then we prove that both $\epsilon_\B$ and its inverse are strict monoidal functors, thus morphisms in $\CBC$. Finally, we show naturality of $\epsilon$.
	
	First, from the snake equations~\eqref{eq:snake} it immediately follows that for any $R \colon X \times Y \to I$ and $S \colon X \to Y$ in $\B$, $\inv\epsilon_\B \epsilon_\B (R)=R$ and $\epsilon_\B \inv\epsilon_\B (S)=S$, so $\epsilon_\B$ and $\epsilon_\B^{-1}$ are inverse of each other as functions on hom-sets.
	
	Next, notice that
	\[
	\inv\epsilon_\B (S) = \delta^{\RA(\B)}_Y \circ (S \times \id_Y) = \Gamma_{\RA(\B)}(S)
	\]
	the last equation holding if and only if $S$ is a map in $\B$ (otherwise we cannot compute $\Gamma_{\RA(\B)}$ on it). This means that $\inv\epsilon_\B$ is an extension of $\Gamma_{\RA(\B)}$ to all of $\B$. Since the symmetry in $\LA\RA(\B)$ is defined as $\Gamma_{\RA(\B)}(\sigma^{\Map(\B)}) = \Gamma_{\RA(\B)} (\sigma^{\B})$ and $\sigma^\B$ is a map in $\B$, we have 
	\[
	\inv\epsilon_\B (\sigma^\B) = \Gamma_{\RA(\B)} (\sigma^\B) = \sigma^{\LA\RA(\B)}.
	\]
	The same can be said for the comultiplication and the counit of every object $X$ in $\B$. Applying $\epsilon_\B$ to both sides of the above equation, one can see that the same holds for $\epsilon_\B$, therefore $\inv\epsilon_\B$ and $\epsilon_\B$ preserve symmetries and the comonoid structure of every object of $\B$ and $\LA\RA(\B)$ respectively. Preservation of the ordering of the hom-sets is guaranteed by the fact that tensor and composition in $\B$ are monotone operations. 
	
	We proceed by showing that $\epsilon_\B$ and $\inv\epsilon_\B$ are strict monoidal functors. 

	Given $R \colon X \times Y \to I$ and $S \colon Y \times Z \to I$ in $\B$, from~\eqref{eq:composition in LARA(B)} we get that
	\[
	\epsilon_\B (s \circ r) = \stikzfig{epsilonScircR} \quad = \stikzfig{epsilonScircEpsilonR} \quad = \epsilon_\B (s) \circ \epsilon_\B(r).
	\]
Preservation of the identity is a direct consequence of the snake equations~\eqref{eq:snake}. Therefore $\epsilon_\B$ is a functor: we need to show that it is strict monoidal. Take now $R \colon X \otimes X' \to I$ and $S \colon Y \otimes Y' \to I$: by~\eqref{eq:composition in LARA(B)}, we have
	\[
	\epsilon_\B(R \otimes S) = {\stikzfig{epsilonFtensorG} = \stikzfig{epsilonFtensorEpsilonG}} = \epsilon_\B (f) \otimes \epsilon_\B(g)
	\]
	where the middle equality is due to the naturality of the unitors and the symmetry of $\B$, as well as the coherence for symmetric monoidal categories.
	
	As per $\inv\epsilon_\B$, from~\eqref{eq:composition in LARA(B)} and the snake equations\eqref{eq:snake} we get
	\[
	\inv\epsilon_\B(S) \circ \inv\epsilon_\B (R) = {\stikzfig{epsilonINVgEpsilonINVf}} \quad 
												 = {\stikzfig{epsilonINVgf}} \quad 
												 = \inv\epsilon_\B(S \circ R)
	\]
	therefore $\inv\epsilon_\B$ preserves composition; preservation of the identity is immediate from the definition of $\epsilon_\B$. Regarding the monoidal product, from~\eqref{eq:tensor in LARA(B)} and naturality of the symmetry of $\B$ we have:
	\[
		\inv\epsilon_\B(R \otimes S) = {\stikzfig{epsilonINVftensorg}} \quad = \quad  {\stikzfig{epsilonINVftensorEpsilonINVg}} = \inv\epsilon_\B(R) \otimes \inv\epsilon_\B (S).
	\]
	This proves that $\epsilon_\B$ and $\inv\epsilon_\B$ are inverse morphisms in $\CBC$. 
	
	All that is left to do is to prove the naturality of $\epsilon$. To that end, let $F \colon \A \to \B$ be a morphism in $\CBC$: we show that the following square
	\[
	\begin{tikzcd}
	\LA\RA(\A) \ar[r,"\epsilon_\A"] \ar[d,"\LA\RA(F)"'] & \A \ar[d,"F"] \\
	\LA\RA(\B) \ar[r,"\epsilon_\B"] & \B
	\end{tikzcd}
	\]
	commutes. We have that $\RA(F)=({F}\restriction_{\Map(\A)}, b^F)$ where for $u \colon A \to I$ in $\A$, $b^F_A(u)=F(u)$. Therefore $\LA\RA(F) \colon \Bicat_{\RA(\A)} \to \Bicat_{\RA(\B)}$ is the functor that maps $R \in \A(X \times Y, I)$ into $F(R) \in \B(FX \times FY, I)$. Hence the square above acts on $R \colon X \times Y \to I$ in $\A$ as follows:
	\[
	{\stikzfig{naturalityEpsilon}}
	\]
	where the right vertical arrow is correct because $F$ preserves all the structure of cartesian bicategory. Thus $\epsilon$ is natural.
\end{proof}

\begin{theorem}\label{thm:adjunction in appendix}
$\LA$ is left adjoint to $\RA$.
\end{theorem}
\begin{proof}
All that is left to do is to show that the triangular equalities~\eqref{eq:triangular equalities diagrams} are satisfied. The first one requires that, for $P \colon \op\C \to \InfSL$ an elementary existential doctrine, the following diagram commute:
	\begin{equation}\label{eq:first triangular equation commutative diagram}
		\begin{tikzcd}
		\LA(P) \ar[r,"\LA(\eta_P)"] \ar[dr,"\id_{\LA(P)}"'] & \LA\RA\LA(P) \ar[d,"\epsilon_{\LA(P)}"] \\
		& \LA(P)
		\end{tikzcd}
	\end{equation}
	Recall that $\LA\RA\LA(P)=\Bicat_{\RA\LA(P)}$, therefore its objects are those of the domain category of the doctrine 
	\[
	\RA\LA(P) = \RA(\Bicat_P) = \Hom_{\Bicat_P}(-,I) \colon \op{\Map(\Bicat_P)} \to \InfSL
	\]
	while its homsets are
	\[
	\LA\RA\LA(P)(X,Y)=\RA\LA(P)(X \times Y) = \Hom_{\Bicat_P}(X \times Y, I) = P(X \times Y \times I).
	\]
	By definition of $\LA$ and $\epsilon$, we have
	\[
	\begin{tikzcd}[row sep=0em]
	\Bicat_P \ar[r,"\LA(\eta_P)"] & \Bicat_{\RA\LA(P)} \\
	X \ar[r,|->] \ar[d,"P(X \times Y) \ni R"'] & X \ar[d,"P_{\rho_{X \times Y}} (R) \in P(X \times Y \times I)"] \\[2em]
	Y \ar[r,|->] & Y  
	\end{tikzcd}
	\qquad
	\begin{tikzcd}[row sep=0em]
	\Bicat_{\RA\LA(P)} \ar[r,"\epsilon_{\LA(P)}"] & \Bicat_P \\
	X \ar[r,|->]  \ar[d,"P(X \times Y \times I) \ni Q"'] & X  \ar[d,"\Gamma_P(\inv\rho_X) \seq \bar Q \seq \Gamma(\lambda_Y)"] \\[2em]
	Y \ar[r,|->] & Y
	\end{tikzcd}	
	\]
where $\bar Q$ is given as in~\eqref{eq:epsilon(r)}. Therefore $\epsilon_{\LA(P)} \circ \LA(\eta_P)$ is the identity on objects and its action on $R \in P(X \times Y)$ is
	\begin{equation}\label{eq:first triangular equation}
	\Gamma_P \left({\stikzfig{epsilonFirstBit}} \right) \seq (P_{\rho_{X \times Y}}(R) \otimes \delta_Y) \seq \Gamma_P(\lambda_Y)
	\end{equation}
	where the above composite is performed in $\Bicat_P$.

	In the following calculations, we will have to refer to various instances of the two objects $X$ and $Y$, and we will do so as annotated in the following string diagram in $\Bicat_P$, showing $\epsilon_{\LA(P)} \circ \LA(\eta_P)(R)$:
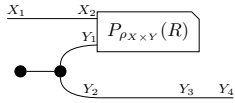
	\[
	\stikzfig{FirstTriangularEquation}
	\]
	In this sense, the first morphism in the chain~\eqref{eq:first triangular equation} goes from $X_1$ to $X_2 \otimes Y_1 \otimes Y_2$ (as a morphism in $\Bicat_P$), the middle one goes from $X_2 \otimes Y_1 \otimes Y_2$ to $I \otimes Y_3$, whereas $\Gamma_P(\lambda_Y)$ goes from $I \otimes Y_3$ to $Y_4$. (Again, the use of subscripts is only a notational convenience to clarify which projections we are taking into consideration in the following computations; as objects of $\C$, we have $X=X_1=X_2$ and $Y=Y_1=Y_2=Y_3=Y_4$.)
	
	We begin by computing the middle morphism:
	\begin{align*}
		P_{\rho_{X \times Y}}(R) \otimes \delta_Y &= P_{X_2Y_1Y_2 I Y_3 \to X_2 Y_1 I} \bigl( P_{\rho_{X \times Y}} (R) \bigr) \land P_{X_2Y_1Y_2 I Y_3 \to Y_2 Y_3} (\delta_Y) \\
		&= P_{X_2 Y_1 Y_2 I Y_3 \to X_2 Y_1} (R) \land P_{X_2 Y_1 Y_2 I Y_3 \to Y_2 Y_3} (\delta_Y)
	\end{align*}
	therefore we have that $(P_{\rho_{X \times Y}}(R) \otimes \delta_Y) \seq \Gamma_P(\lambda_Y)$ is equal to:
	\begin{align*}
		& \exists_{X_2 Y_1 Y_2 I Y_3 Y_4 \to X_2 Y_1 Y_2 Y_4} \Bigl[ P_{X_2 Y_1 Y_2 I Y_3 Y_4 \to X_2 Y_1 Y_2 I Y_3} \bigl( P_{X_2 Y_1 Y_2 I Y_3 \to X_2 Y_1} (R) \\
		& \hspace{21em}\land P_{X_2 Y_1 Y_2 I Y_3 \to Y_2 Y_3} (\delta_Y)  \bigr) \\
		& \hspace{11em} \land P_{X_2 Y_1 Y_2 I Y_3 Y_4 \to I Y_3 Y_4} \bigl( P_{\lambda_Y \times \id_Y} (\delta_Y) \bigr)  \Bigr] \\
		&= \exists_{X_2 Y_1 Y_2 I Y_3 Y_4 \to X_2 Y_1 Y_2 Y_4} \Bigl[ P_{X_2 Y_1 Y_2 I Y_3 Y_4 \to X_2 Y_1} (R) \\ 
		& \hspace{12em} \land P_{X_2 Y_1 Y_2 I Y_3 Y_4 \to Y_2 Y_3} (\delta_Y) \land P_{X_2 Y_1 Y_2 I Y_3 Y_4 \to Y_3 Y_4} (\delta_Y) \Bigr] \displaybreak[0]\\
		&=\exists_{X_2 Y_1 Y_2 I Y_3 Y_4 \to X_2 Y_1 Y_2 Y_4} \Bigl[ P_{X_2 Y_1 Y_2 I Y_3 Y_4 \to X_2 Y_1} (R) \\
		& \hspace{12em} \land \exists_{X_2 Y_1 Y_2 I Y_3 \to X_2 Y_1 Y_2 I Y_3 Y_3} \bigl( P_{X_2 Y_1 Y_2 I Y_3 \to Y_2 Y_3} (\delta_Y) \bigr) \Bigr] \displaybreak[0]\\
		&=\exists_{X_2 Y_1 Y_2 I Y_3 Y_4 \to X_2 Y_1 Y_2 Y_4} \Bigl[ P_{X_2 Y_1 Y_2 I Y_3 Y_4 \to X_2 Y_1 Y_2 Y_4} \bigl( P_{X_2 Y_1 Y_2 Y_4 \to X_2 Y_1} (R) \bigr) \\
		& \hspace{12em} \land  \exists_{X_2 Y_1 Y_2 I Y_3 \to X_2 Y_1 Y_2 I Y_3 Y_3} \bigl( P_{X_2 Y_1 Y_2 I Y_3 \to Y_2 Y_3} (\delta_Y) \bigr) \Bigr] \displaybreak[0]\\
		&= P_{X_2 Y_1 Y_2 Y_4 \to X_2 Y_1} (R) \\
		& \hspace{1em}\land \exists_{X_2 Y_1 Y_2 I Y_3 Y_4 \to X_2 Y_1 Y_2 Y_4} \Bigl[ \exists_{X_2 Y_1 Y_2 I Y_3 \to X_2 Y_1 Y_2 I Y_3 Y_3} \bigl( P_{X_2 Y_1 Y_2 I Y_3 \to Y_2 Y_3} (\delta_Y) \bigr) \Bigr] \displaybreak[0]\\
		&= P_{X_2 Y_1 Y_2 Y_4 \to X_2 Y_1} (R) \land \exists_{X_2 Y_1 Y_2 I Y_3 \to X_2 Y_1 Y_2 Y_3} \bigl( P_{X_2 Y_1 Y_2 I Y_3 \to Y_2 Y_3} (\delta_Y) \bigr) \\
		&= P_{X_2 Y_1 Y_2 Y_4 \to X_2 Y_1} (R) \land  P_{X Y_1 Y_2 Y_3 \to Y_2 Y_3} (\delta_Y).
	\end{align*}
	The last equation above is due to the fact that the morphism $e=X_2 Y_1 Y_2 I Y_3 \to X_2 Y_1 Y_2 Y_3$ is an isomorphism and therefore the left adjoint of $P_e$, which is $\exists_e$, is actually $P_{\inv e}$. Hence, the composite~\eqref{eq:first triangular equation} is equal to:
	\begin{align*}
		&\exists_{X_1 X_2 Y_1 Y_2 Y_4 \to X_1 Y_4} \Bigl[ P_{X_1 X_2 Y_1 Y_2 Y_4 \to X_1 X_2} (\delta_X) \land P_{X_1 X_2 Y_1 Y_2 Y_4 \to X_2 Y_1} (R) \\
		& \hspace{8.5em} \land P_{X_1 X_2 Y_1 Y_2 Y_4 \to Y_1 Y_2} (\delta_Y) \land P_{X_1 X_2 Y_1 Y_2 Y_4 \to Y_2 Y_4} (\delta_Y) \Bigr] \\
		&=\exists_{X_1 X_2 Y_1 Y_2 Y_4 \to X_1 Y_4} \Bigl[ \exists_{X Y_1 Y_2 Y_4 \to X X Y_1 Y_2 Y_4} \bigl( P_{X Y_1 Y_2 Y_4 \to X Y_1} (R) \bigr) \\
		& \hspace{10em} \land P_{X_1 X_2 Y_1 Y_2 Y_4 \to Y_1 Y_2} (\delta_Y) \land P_{X_1 X_2 Y_1 Y_2 Y_4 \to Y_2 Y_4} (\delta_Y) \Bigr] \displaybreak[0]\\
		& =\exists_{X_1 Y_2 Y_4 \to X_1 Y_4} \\
		&\hspace{1.35em}\exists_{X_1 Y_1 Y_2 Y_4 \to X_1 Y_2 Y_4} \\
		& \hspace{1.35em}\exists_{X_1 X_2 Y_1 Y_2 Y_4 \to X_1 Y_1 Y_2 Y_4} \Bigl[ P_{X_1 X_2 Y_1 Y_2 Y_4 \to X_1 Y_1 Y_2 Y_4} \bigl( P_{X_1 Y_1 Y_2 Y_4 \to Y_1 Y_2} (\delta_Y) \\
		&\hspace{22em} \land P_{X_1 Y_1 Y_2 Y_4 \to Y_2 Y_4} (\delta_Y)  \bigr) \\
		& \hspace{12em}\land \exists_{X Y_1 Y_2 Y_4 \to X X Y_1 Y_2 Y_4} \bigl( P_{X Y_1 Y_2 Y_4 \to X Y_1} (R) \bigr) \Bigr] \displaybreak[0]\\
		& =\exists_{X_1 Y_2 Y_4 \to X_1 Y_4} \\
		&\hspace{1.35em} \exists_{X_1 Y_1 Y_2 Y_4 \to X_1 Y_2 Y_4} \Bigl[
			 P_{X_1 Y_1 Y_2 Y_4 \to Y_1 Y_2} (\delta_Y) \land P_{X_1 Y_1 Y_2 Y_4 \to Y_2 Y_4} (\delta_Y) \\
		&\hspace{10em} \land \exists_{X_1 X_2 Y_1 Y_2 Y_4 \to X_1 Y_1 Y_2 Y_4} \\
		&\hspace{11.1em} \exists_{X Y_1 Y_2 Y_4 \to X X Y_1 Y_2 Y_4} \bigl( P_{X Y_1 Y_2 Y_4 \to X Y_1} (R) \bigr) \Bigr] \displaybreak[0]\\ 
		&= \exists_{X Y_2 Y_4 \to X Y_4} \\
		& \hspace{1.35em} \exists_{X Y_1 Y_2 Y_4 \to X Y_2 Y_4} \Bigl[ P_{X Y_1 Y_2 Y_4 \to X Y_2 Y_4} \bigl(P_{X Y_1 Y_2 Y_4 \to Y_2 Y_4} (\delta_Y) \bigr) \\
		& \hspace{9em} \land P_{X Y_1 Y_2 Y_4 \to Y_1 Y_2} (\delta_Y) \land P_{X Y_1 Y_2 Y_4 \to X Y_1} (R) \Bigr]	\displaybreak[0]\\
		&=\exists_{X Y_2 Y_4 \to X_1 Y_4} \Bigl[ P_{X Y_1 Y_2 Y_4 \to Y_2 Y_4} (\delta_Y) \\
		& \hspace{7em} \land \exists_{X Y_1 Y_2 Y_4 \to X Y_2 Y_4} \exists_{X Y_1 Y_4 \to X Y_1 Y_1 Y_4} \bigl( P_{X Y_1 Y_4 \to X Y_1} (R) \bigr)
		\Bigr]		\displaybreak[0]\\
		&= \exists_{X Y_2 Y_4 \to X_1 Y_4} \Bigl[ P_{X Y_1 Y_2 Y_4 \to Y_2 Y_4} (\delta_Y) \land P_{X Y_1 Y_4 \to X Y_1} (R) \Bigr] \displaybreak[0]\\
		&= \exists_{X Y_2 Y_4 \to X_1 Y_4} \exists_{X Y \to X Y Y} (R) \\
		&= R.
	\end{align*}
	This proves that triangle~\eqref{eq:first triangular equation commutative diagram} commutes. 
	
	The second triangular equality requires that, for $\B$ a cartesian bicategory, the following diagram commute:
	\begin{equation}\label{eq:second triangular equality}
	\begin{tikzcd}
		\RA(\B) \ar[r,"\eta_{\RA(\B)}"] \ar[dr,"\id_{\RA(\B)}"'] & \RA\LA\RA(\B) \ar[d,"\RA(\epsilon_\B)"] \\
		& \RA(\B)
		\end{tikzcd}
	\end{equation}
	Recall that $\eta_{\RA(\B)} = (\Gamma_{\RA(B)}, \RA(\B)\rho)$, where $\Gamma_{\RA(B)} \colon \Map(\B) \to \LA\RA(\B)$ coincides with the restriction of $\inv\epsilon_B$ (see~\eqref{eq:epsilon(r)}) to $\Map(\B)$ as shown in the proof of Lemma~\ref{lemma:epsilon natural isomorphism}, while
	\[
	(\RA(\B)\rho)_X = \Hom_{\B}(-,I)(\rho_X) = - \circ \rho_X \colon \Hom_{\B}(X,I) \to \Hom_{\B}(X \otimes I, I)
	\]
	as $\rho_X \colon X \otimes I \to I$ is the right unitor of $\B$. 
	
	Now,
	\[
	\RA\LA\RA(\B) = \Hom_{\LA\RA(\B)} (-,I) = \Hom_{\B}(- \otimes I, I) \colon \Map(\LA\RA(\B)) \to \InfSL
	\]
	and
	\[
	\RA(\epsilon_\B) = (\epsilon_B\restriction_{\Map (\LA\RA(\B))}, b^{\epsilon_\B})
	\]
	where by~\eqref{eq:functor R transformation component} we have:
	\[
	\begin{tikzcd}[row sep=0em]
	\Hom_{\B}(X \otimes I, I) \ar[r,"b^{\epsilon_\B}_X"] & \Hom_{\B}(X, I) \\
	R \ar[r,|->] & \epsilon_\B(R)
	\end{tikzcd}
	\]
	and it it easy to see that, for $R \colon X \otimes I \to I$ in $\B$, $\epsilon_\B (R)= R \circ \inv\rho_X$.
	
	Composition in $\EED$ is component-wise: on the first component of $\eta_{\RA(\B)} \seq \RA(\epsilon_\B)$ we have
	\[
	\begin{tikzcd}[row sep=0em]
	\Map(\B) \ar[r,"\Gamma_{\RA(B)}"] & \Map(\LA\RA(\B)) \ar[r,"\epsilon_B\restriction_{\Map (\LA\RA(\B))}"] & \Map(\B) \\
	f \ar[r,|->] & \inv\epsilon_\B(f) \ar[r,|->] & \epsilon_B\inv\epsilon_\B(f) = f
	\end{tikzcd}
	\]
	while on the second component we have
	\[
	\begin{tikzcd}[row sep=0em]
	\Hom_{\B}(X,I) \ar[r,"(\RA(\B)\rho)_X"] & \Hom_{\B}(X \otimes I, I) \ar[r,"b^{\epsilon_\B}_X"] & \Hom_\B (X,I) \\
	R \ar[r,|->] & R \circ \rho_X \ar[r,|->] & R \circ \rho_X \circ \inv\rho_X = R.
	\end{tikzcd}
	\]
	Thus triangle~\eqref{eq:second triangular equality} commutes. 
	\end{proof}

\section{Appendix to \S~\ref{sec:equivalence}}

We first characterise the maps of $\Bicat_P$ in Proposition~\ref{prop:maps of A_P}, which we need to prove that the Rule of Unique Choice is equivalent to the fullness of the graph functor in  Proposition~\ref{prop:RUC iff Gamma full}. The fact that a doctrine $P$ has comprehensive diagonals if and only if $\Gamma_P$ is faithful is stated in~\cite[Theorem 3.2]{maietti_triposes_2017}. Finally, we give a proof of Theorem~\ref{thm:EEDff and CBC are equivalent}.

\begin{proposition}[{Cf.~\cite{maietti_triposes_2017}}]\label{prop:maps of A_P}
Let $P \colon \op\C \to \InfSL$ be an elementary existential doctrine, $R \in P(X \times Y)$. Then $R$ is a map in $\Bicat_P$ if and only if
\begin{gather}
    \exists_{\pi_1 \colon X Y \to X} (R) = \top_{PX} \label{eq:total relation A_P} \\
    P_{\langle \pi_1,\pi_2 \rangle \colon X YY \to XY} (R) \land P_{\langle \pi_1,\pi_3 \rangle \colon XYY \to XY} (R) \le P_{\langle \pi_2, \pi_3 \rangle \colon XYY \to YY} (\delta_Y).\label{eq:functional relation A_P}
\end{gather}
\end{proposition}
\begin{proof}
Equation~\eqref{eq:total relation A_P} stems from calculating the two sides of 
\[
\sxtikzfig{StrictDis}
\]
while~\eqref{eq:functional relation A_P} is equivalent to
\[
P_{X Y_1 Y_2 \to X Y_1} (R) \land P_{X Y_1 Y_2 \to Y_1 Y_2} (\delta_Y) = P_{X Y_1 Y_2 \to X Y_1} (R) \land P_{X Y_1 Y_2 \to X Y_2} (R)
\]
which is the result of computing the two sides of the equation
\[
\sxtikzfig{StrictCopy}
\]
in the category $\Bicat_P$.
\end{proof}

\begin{proposition}\label{prop:RUC iff Gamma full}
An elementary existential doctrine $P$ satisfies the Rule of Unique Choice if and only if its graph functor $\Gamma_P$ is full.
\end{proposition}
\begin{proof}
We show that for $P \colon \op\C \to \InfSL$, $R \in P(X \times Y)$ map in $\Bicat_P$ and $h \colon X \to Y$ in $\C$:
\[
\top_{PX} \le P_{\langle \id_X,h \rangle} (R) \iff R = P_{h \times \id_Y}(\delta_Y).
\]
First, if $R=P_{h \times \id_Y}(\delta_Y)$, then
\[
    P_{\langle \id_X,h \rangle} (R) = P_{\langle \id_X,h \rangle} P_{h \times \id_Y}(\delta_Y) =P_{\langle h, h \rangle} (\delta_Y) = P_h P_{\Delta_Y} (\delta_Y) = P_h (\top_{PY}) = \top_{PX}
    \]
because $\delta_Y=\exists_{\Delta_Y}(\top_{PY})$ and, since $\exists_\Delta$ is left adjoint to $P_\Delta$, $\top_{PY} \le P_{\Delta_Y} \exists_{\Delta_Y} (\top_{PY})$. 
    Vice versa, suppose that $\top_{PX} \le P_{\langle \id_X,h \rangle} (R)$. Then $P_{XY \to X} (\top_{PX}) = P_{XY \to X} P_{\langle \id_X,h \rangle} (R)$. Now, since the following two diagrams commute:
    \[
    \begin{tikzcd}
    X \times Y \ar[r,"\pi_1"] \ar[d,"{\langle \id_X, h \rangle \times \id_Y}"'] & X \ar[d,"{\langle \id_X, h \rangle}"] \\
    X \times Y \times Y \ar[r,"{\langle \pi_1, \pi_2 \rangle}"'] & X \times Y
    \end{tikzcd}
    \qquad
    \begin{tikzcd}
    A \times B \ar[rr,"h \times \id_Y"] \ar[dr,"{\langle \id_X, h \rangle \times \id_Y}"'] & & X \times Y \\
    & X \times Y \times Y \ar[ur,"{\langle\pi_2, \pi_3 \rangle}"']
    \end{tikzcd}
    \]
    we have:
    \begin{align*}
    R &= R \land \top_{P(XY)} = R \land P_{XY \to X} P_{\langle \id_X, h \rangle} (R) \\
    &=P_{\langle \id_X, h \rangle \times \id_Y} \bigl( P_{X Y_1 Y_2 \to X Y_2} (R) \land P_{X Y_1 Y_2 \to X Y_1} (R) \bigr) \\
    &\le P_{\langle \id_X, h \rangle \times \id_Y} \bigl( P_{X Y_1 Y_2 \to Y_1 Y_2} (\delta_Y) \bigr) = P_{h \times \id_Y} (\delta_Y)
    \end{align*}
    where the inequality above is due to~\eqref{eq:functional relation A_P}. So we proved that $R \le P_{h \times \id_Y}(\delta_Y)$, while for the other inequality:
    \begin{align*}
        P_{h \times \id_Y}(\delta_Y) &= P_{h \times \id_Y}(\delta_Y) \land \top_{P(XY)} \\
        &= P_{h \times \id_Y}(\delta_Y) \land P_{XY \to X} P_{\langle \id_X, h \rangle} (R) \\
        &= P_{\langle \id_X, h \rangle \times \id_Y}  \bigl( \exists_{\id_X \times \Delta_Y} (R) \bigr)
    \end{align*}
    because of~\eqref{eq:elementary2}. Since $R = P_{\id_X \times \Delta_Y} P_{X Y_1 Y_2 \to X Y_2} (R)$, by adjunction we have $\exists_{\id_X \times \Delta_Y}(R) \le P_{X Y_1 Y_2 \to X Y_2} (R)$, hence we get
    \[
    P_{h \times \id_Y}(\delta_Y) \le P_{\langle \id_X, h \rangle \times \id_Y} P_{X Y_1 Y_2 \to X Y_2} (R) = R. \qedhere
    \]
    \end{proof}

\paragraph*{Proof of Theorem~\ref{thm:EEDff and CBC are equivalent}}
By restricting $\LA$ and corestricting $\RA$ to $\EEDff$, one obtains two functors $\overline\LA$ and $\overline\RA$ and we can define two natural transformations $\overline\eta$ and $\overline\epsilon$ whose individual components are the same as $\eta$ and $\epsilon$ respectively. This means that $\overline\eta$ and $\overline\epsilon$ are still natural transformations satisfying the triangular equalities~\eqref{eq:triangular equalities diagrams}, therefore $\overline\LA \dashv \overline\RA$, with $\overline\epsilon$ being a natural isomorphism. 

Let $P \colon \C \to \InfSL$ be in $\EEDff$. Then $\Gamma_P$ is faithful because of~\cite[Theorem 3.2]{maietti_triposes_2017} and full because of Proposition~\ref{prop:RUC iff Gamma full}. Hence $\overline\eta_P=(\Gamma_P, P\rho)$ has an inverse in $\EEDff$, namely $(\inv\Gamma_P, P\inv\rho)$, where $\inv\Gamma_P$ is the identity on objects and sends every map $R \colon X \to Y$ in $\RA\LA(P)$ into the unique morphism $f \colon X \to Y$ of $\C$ such that $\Gamma_P(f)=R$. Therefore also $\overline\eta$ is a natural isomorphism, hence $(\overline\LA, \overline\RA, \overline\eta, \overline\epsilon)$ is an adjoint equivalence. \qed 
\end{document}